\documentclass[11pt]{article}

\usepackage{amssymb,amsmath,amsthm}	% Mathematische Symbole
\usepackage{a4wide}
\usepackage{graphicx}      										% Grafiken einbinden
\usepackage{tikz}          										% LaTeX-Grafiken erzeugen
\usepackage{url}
\usepackage{dsfont}
\usepackage{amsmath}
\usepackage{xspace}

\newtheorem{definition}{Definition}[]

\newtheorem{theorem}{Theorem}[]
\newtheorem{remark}{Remark}[]
\newtheorem{corollary}{Corollary}[]

\newtheorem{proposition}{Proposition}[]

\newcommand{\R}{\ensuremath{\mathbb{R}}}

\begin{document}
	
	%-- TITEL ----------------------------------------------------------%
	\title{From Disequilibrium Markets to Equilibrium}
	\author{Christian Lax\footnote{Lehrstuhl A, RWTH Aachen, Templergraben 55, 52056 Aachen, Germany }, Torsten Trimborn\footnote{IGPM, RWTH Aachen, Templergraben 55, 52056 Aachen, Germany} \footnote{Corresponding author: trimborn@igpm.rwth-aachen.de}}

	\maketitle

   %-- INHALTSVERZEICHNIS ----------------------------------------------------------%

\begin{abstract}
The modeling of financial markets as disequilibrium models by ordinary differential equations has become a popular modeling tool. 
One famous example of such a model is the  Beja-Goldman model \cite{beja1980dynamic} which we consider in this paper. 
We study the passage from disequilibrium dynamics to equilibrium. Mathematically, this limit corresponds to an asymptotic limit also known as 
a Tikhonov-Fenichel reduction. Furthermore, we analyze the stability of the reduced equilibrium model and discuss the economic implications. 
 We conduct several numerical examples to visualize and support our analysis. 
\end{abstract}	
{\textbf{Keywords:} Beja-Goldmann Model, disequilibrium, equilibrium, rational market, asymptotic limit, Tikhnov-Fenichel reduction, high frequency trader }

\section{Introduction }	
In the past decades the interest in rational markets built on the general equilibrium theory has shifted to irrational markets also known as disequilibrium models. 
General equilibrium theory dates back to the early works of Walras \cite{walras2013elements} and has been further developed in the last century by McKenzie, Arrow and Debreu  \cite{walker2006walrasian}.
Heuristically speaking, an equilibrium price is reached when supply matches demand. This equilibrium theory in particular assumes that we have a rational market meaning that
there are no transaction costs and perfect informations. Many major contributions in finance such as the portfolio theory by Markowitz \cite{markowitz1952portfolio} and Merton \cite{merton1969lifetime} or the capital asset pricing model by Sharpe \cite{sharpe1964capital} and Lintner \cite{lintner1965security} are built on the rational market hypothesis. \\ \\
Mainly based on the restricted nature of the assumptions of a rational markets, the general equilibrium theory has been critizised \cite{beja1977orders, beja1980dynamic, heertje2002recent, ackerman2002still}. This has lead to the theory of disequilibrium markets \cite{beja1980dynamic,chiarella1986perfect,  he2011dynamic, day1990bulls} probably first introduced by Beja and Goldman \cite{beja1980dynamic}. Major contributions in the field of agent-based models are build on the idea of disquilibrium markets \cite{frankel1990chartists, day1990bulls, chiarella1992dynamics, chiarella1992developments, kirman1993ants, lux1995herd, hommes2006heterogeneous}. Before we are able to present the general idea of disequlibrium markets we have to introduce the notion of aggregated excess demand \cite{mantel1974characterization, debreu1974excess, sonnenschein1972market}.
The excess demand denotes the aggregated supply and demand of all financial agents. More precisely the excess demand $ED$ is defined as the sum of agents' supply subtracted from agent's demand. 
For an equilibrium price $P^*$ 
\begin{align}\label{ED}
ED(P^*)=0
\end{align}
has to hold. The price $P$ denotes the logarithmic stock price of an asset. More generally one can even consider $H(ED)=0$ for a nonlinear function $H$, zero at the origin as suggested by several authors  \cite{campbell1997econometrics, kempf1999market, cont2000herd}. Nevertheless, the linear form \eqref{ED} can be seen as linearization of the nonlinear form $H(ED)=0$ as argued in \cite{beja1980dynamic, SABCEMM}. 
The general form of a disequilibrium market model is then given by
\begin{align}\label{DisModel}
\frac{d}{dt} P(t) =  \frac{1}{\epsilon}\ ED(P(t)).
\end{align}
The constant $\frac{1}{\epsilon}$ denotes the market depth or the speed of price adjustment \cite{kempf1999market}.  In comparison to the equlibirum market model \eqref{ED} the market depth is finite. 
Mathematically speaking, the disequilibirum model \eqref{DisModel} is a relaxation of the algebraic relation \eqref{ED} with the relaxation parameter $\epsilon$. In this paper we study the asymptotic limit $\epsilon\to 0$ for disequilibrium models which consist of price equations of the type \eqref{DisModel} coupled to an additional ODE. Thus, we study the limit from an disequilibirum model to an equilibirum model.  \\ \\
Such asymptotic limits are also known in the case of a singular perturbation problem as Tikhonov-Fenichel reductions \cite{fenichel1979geometric, tikhonov1952systems, goeke2014constructive, Hoppensteadt}. This theory for dynamical systems has been mainly applied to chemical reaction kinetics \cite{heineken1967mathematical, lax2018singular, frank2018quasi}. 
The advantage of this asymptotic limit is to obtain a possibly much simpler reduced form of the original dynamical system. Such a reduced system can be conveniently analyzed by classical tools and is a good approximation of the original dynamics for small $\epsilon$. \\ \\
In this paper we consider two different asymptotic limits of the Beja-Goldman model. The Beja-Goldman model is a two dimensional dynamical system. The price equation is of the form \eqref{DisModel} and is coupled 
to the time evolution of the chartist estimate. More precisely, the aggregated excess demand is given by the excess demand of two representative agents,  chartist and fundamentalists. 
 First we study the so called liquid market limit which corresponds to an infinite large market depth. The reduced model can be seen as the equilibrium market version of the Beja-Goldman model. 
 Secondly we study the liquid chartist limit which can be seen as the limit of an infinite fast reaction speed of chartists. One may characterize this limit as a disequilibrium market model with high-frequency trader. 
The resulting reduced models are one dimensional ordinary differential equations coupled to an algebraic equation, which we can easily analyze.\\ \\
The outline of the paper is as follows. In the next section we introduce the Beja-Goldman model in detail and present the dynamical behavior. In section 3 we give an introduction to singular perturbation problems and the Tikonov-Fenichel reduction. Then we derive the reduced Beja-Goldman model in the liquid market and liquid chartist limit, present numerical experiments and analyze the obtained reduced models. 
In section 4 we give an economic interpretation of the results and a small conclusion of this work.

\clearpage

\section{The Beja-Goldman Model}
We introduce the Beja-Goldman model \cite{beja1980dynamic} and present the dynamical behavior of the model, as studied in \cite{beja1980dynamic}. 
The Beja-Goldman model with parameters $a>0,b> 0, r>0, F>0$ and scaling parameter (or relaxation parameter) $ \gamma,\epsilon>0$ is given by
\begin{subequations} \label{BGO}
\begin{align}
&\dot{P}(t)= \frac{1}{\epsilon} \left( a[ F-P(t)]+ b[\Psi(t) -r]  \right),\\
&\dot{\Psi}(t) = \frac{1}{\gamma} (\dot{P}(t)-\Psi(t)).
\end{align}
\end{subequations}
Here, $P\in\R$ denotes the logarithmic stock price and $\Psi\in\R$ so called chartist' price estimate. The parameter $r$ denotes the bond return and $F$ the fundamental price. The aggregated excess demand is defined as the sum of the fundamental ($ed^f$) and  chartists' ($ed^c$) demands:
\begin{align*}
&ed^f(t) := a\ (F-P(t)),\\
&ed^c(t) := b\ (\Psi(t)-r).
\end{align*}
The parameter $a$ respectively $b$ denote the market power of fundamentalists respectively chartists. This modeling approach of chartist and fundamental demand dates back to the work by Zeeman \cite{zeeman1974unstable} and has been frequently used in agent based modeling \cite{levy2000microscopic, chiarella2006asset, brock1997rational, brock1998heterogeneous, chiarella2006asset, franke2012structural}.  The scaling parameter $\epsilon$ denotes the inverse market depth. The chartists' estimate is defined as an relaxation of the price change. The parameter $\gamma$ is the inverse reaction speed of chartists.\\ 
We can rewrite the Beja-Goldman model \eqref{BGO} as follows:
\begin{subequations}\label{BG}
\begin{align}
&\dot{P}(t)= \frac{1}{\epsilon} \left( a[ F-P(t)]+ b[\Psi(t) -r]  \right),\label{PriceODE}\\ 
&\dot{\Psi}(t) = \frac{1}{\gamma} ([ \frac{1}{\epsilon} \left( a[ F-P(t)]+ b[\Psi(t) -r]  \right)]-\Psi(t)).\label{ChartistODE}
\end{align}
\end{subequations}
The matrix form of our system \eqref{BG} reads.
$$
\dot{\boldsymbol{X}} = \boldsymbol{A}\boldsymbol{X}+\boldsymbol{B}
$$
with
$$\boldsymbol{X}:=\begin{pmatrix}
 P\\ \Psi
\end{pmatrix} ,\quad 
 \boldsymbol{A}:=
\begin{pmatrix}
-\frac{a}{ \epsilon} & \frac{b}{\epsilon}\\
-\frac{a}{\gamma\ \epsilon} & -\frac{1}{\gamma}\ (1-\frac{b}{\epsilon})
\end{pmatrix},\quad \boldsymbol{B}:= \begin{pmatrix}    \frac{1}{\epsilon}\ (a\ F- b\ r)\\  \frac{1}{\epsilon \gamma}\ (a\ F-b\ r)\end{pmatrix}.$$
Mathematically, our model is an inhomogeneous linear differential system. The eigenvalues of $\boldsymbol{A}$ are given by
\begin{align*}
&\lambda_1:=\frac{-a\ \gamma+b-\epsilon-\sqrt{(a\gamma-b+\epsilon)^2-4\ a\ \gamma\ \epsilon}}{2\ \gamma\ \epsilon},\\
& \lambda_2:=\frac{-a\ \gamma+b-\epsilon+\sqrt{(a\gamma-b+\epsilon)^2-4\ a\ \gamma\ \epsilon}}{2\ \gamma\ \epsilon}
\end{align*}
and the eigenvectors are:
\begin{align*}
& \boldsymbol{v}_1:=\begin{pmatrix}\frac{a\ \gamma+b-\epsilon+\sqrt{(a\gamma-b+\epsilon)^2-4\ a\ \gamma\ \epsilon}}{2\ a}\\ 1 \end{pmatrix}, \quad \quad 
 \boldsymbol{v}_2:=\begin{pmatrix}\frac{a\ \gamma+b-\epsilon-\sqrt{(a\gamma-b+\epsilon)^2-4\ a\ \gamma\ \epsilon}}{2\ a}\\ 1 \end{pmatrix}.
\end{align*}

\begin{proposition}
The solution basis of the ODE
$$
\dot{\boldsymbol{Y}} = A\ \boldsymbol{Y},
$$
is given by
\begin{align*}
&\Phi_1: \R \to \R^2, t\mapsto e^{\lambda_1\ t}\  \boldsymbol{v}_1,\quad \quad \Phi_2: \R \to \R^2, t\mapsto e^{\lambda_2\ t} \ \boldsymbol{v}_2,\\
\end{align*}
in the case of real eigenvalues $\lambda_1,\lambda_2$. In the case of complex eigenvalues
$\lambda_2=\bar{\lambda}_1$ the solution basis reads:
\begin{align*}
&\Phi_1: \R \to \R^2, t\mapsto Re\Big( e^{\lambda_1\ t} \ \boldsymbol{v}_1\Big),\quad\quad \Phi_2: \R \to \R^2, t\mapsto Im\Big( e^{\lambda_2\ t} \ \boldsymbol{v}_2\Big).\\
\end{align*}
\end{proposition}
We define the fundamental matrix of the Beja-Goldman model by $\Phi(t):=(\Phi_1, \Phi_2)\in\R^{2\times 2}$. 
Then, the solution of the Beja-Goldman model is given by
$$
\boldsymbol{X} = \Phi(t)\ \Phi(0)^{-1}\ \boldsymbol{X}(0) +\Phi(t)\  \int\limits_0^t \Phi(s)^{-1} \boldsymbol{B}(s)\ ds.
$$
We want to recap the stability results of the Beja-Goldman model as discussed in \cite{beja1980dynamic}. 
\begin{proposition}\label{Stability}
The system \eqref{BG} is stable if and only if $a>\frac{1}{\gamma}(b-\epsilon)$.\\
The system \eqref{BG} is oscillatory if and only if $(\sqrt{\epsilon}-\sqrt{a\ \gamma})^2< b<(\sqrt{\epsilon}+\sqrt{a\ \gamma})^2$.
\end{proposition}

\begin{figure}[h!]
\begin{center}
\includegraphics[width=0.4\textwidth]{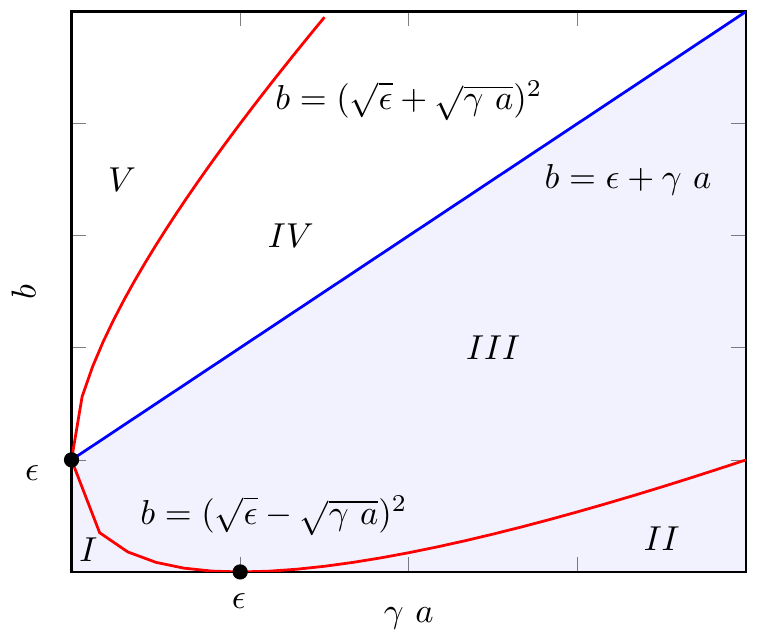}
\end{center}
\caption{Stability region with respect to the parameters $b$ and $\gamma\ a$. The colored area denotes the stable set. The oscillations occur in the area between the red curves denoted by $III$ and $IV$. (based on Figure 1 in \cite{beja1980dynamic}) }\label{Regions}
\end{figure}

We aim to visualize the different limit behavior of the price and chartist estimate as analyzed in Proposition \ref{Stability}. For sufficiently large market power of the fundamentalists it is always possible to obtain stable dynamics. More precisely stable dynamics lead to convergence of the chartist estimate to zero $\Psi_{\infty}= 0$ and to the equilibrium price $P_{\infty}:= F-\frac{b}{a}r$. This behavior can be obtained for the regions $I)-III)$, as defined in Figure \ref{Regions}, see Figure \ref{StabDyn} and Figure \ref{OscDyn}. Increasing the market power of chartists leads to oscillatory stable behavior (region $III$), then to oscillatory unstable behavior (region $IV$) and finally to a blow up (region $V$). This dynamic is depicted in the Figures \ref{OscDyn} and \ref{BUDyn}. Furthermore, we have an example of the dynamics at the stability border $b=\epsilon +\gamma\ a$ (see Figure \ref{BUDyn}).

\begin{figure}[h!]
\begin{center}
\includegraphics[width=0.45\textwidth]{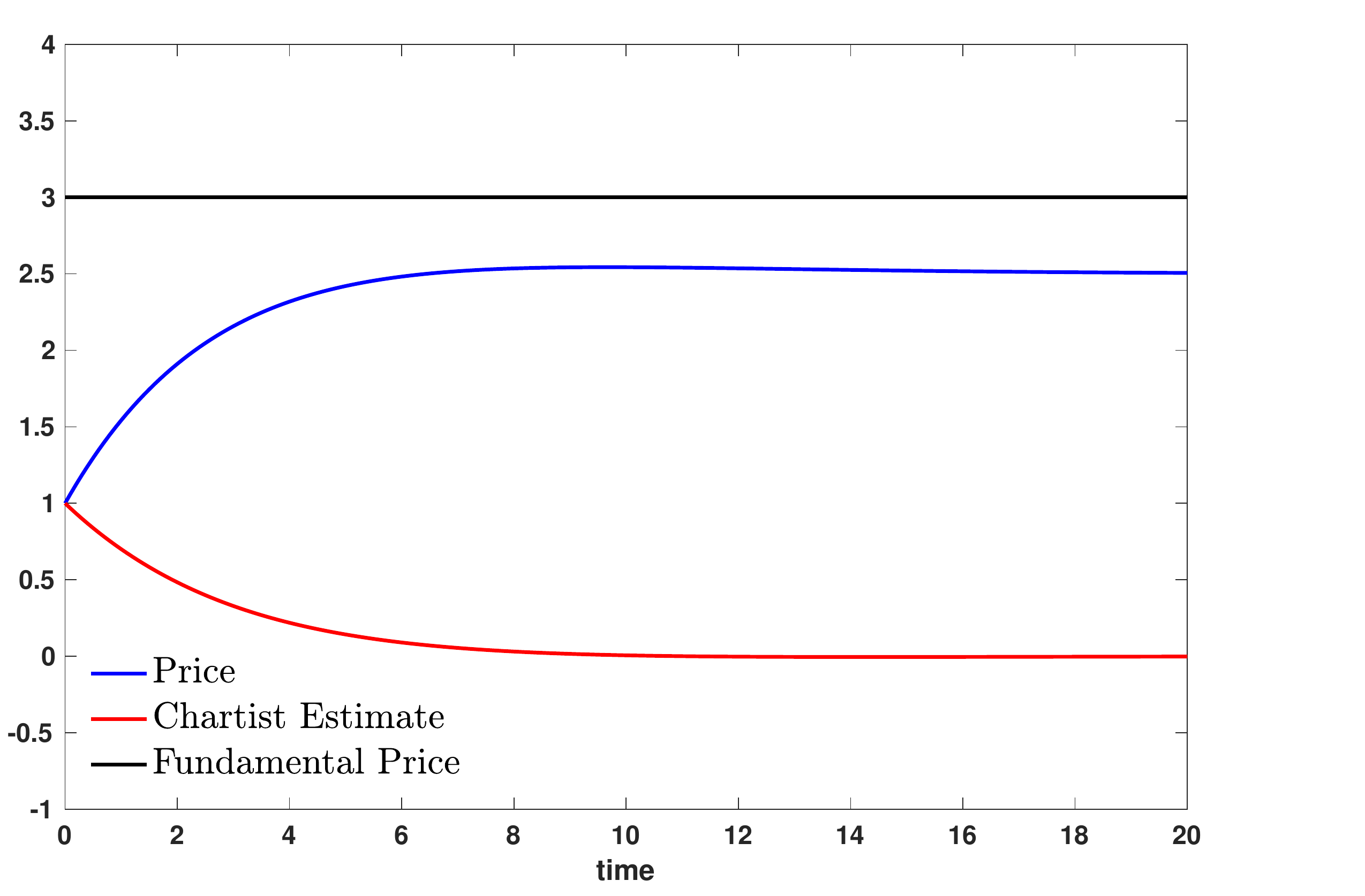}
\hfill
\includegraphics[width=0.45\textwidth]{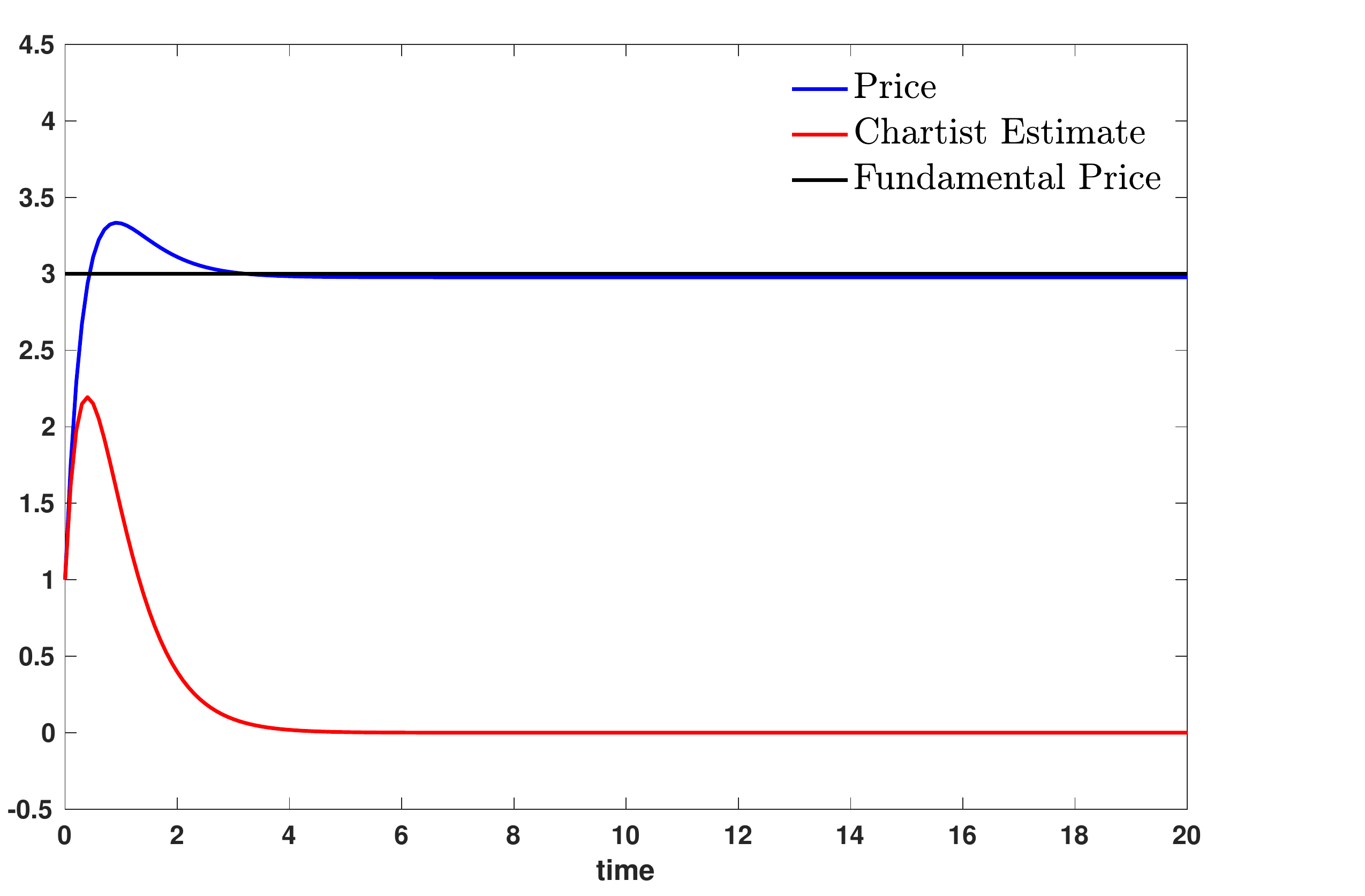}
\end{center}
\caption{LHS: Stability Region I with parameters $a=0.1,\ b=0.5,\ \gamma=1,\ \epsilon=1,\ r=0.1$.
RHS: Stability Region II with parameters $a=4,\ b=0.9,\ \gamma=1,\ \epsilon=1,\ r=0.1$. }\label{StabDyn}
\end{figure}

\begin{figure}[h!]
\begin{center}
\includegraphics[width=0.45\textwidth]{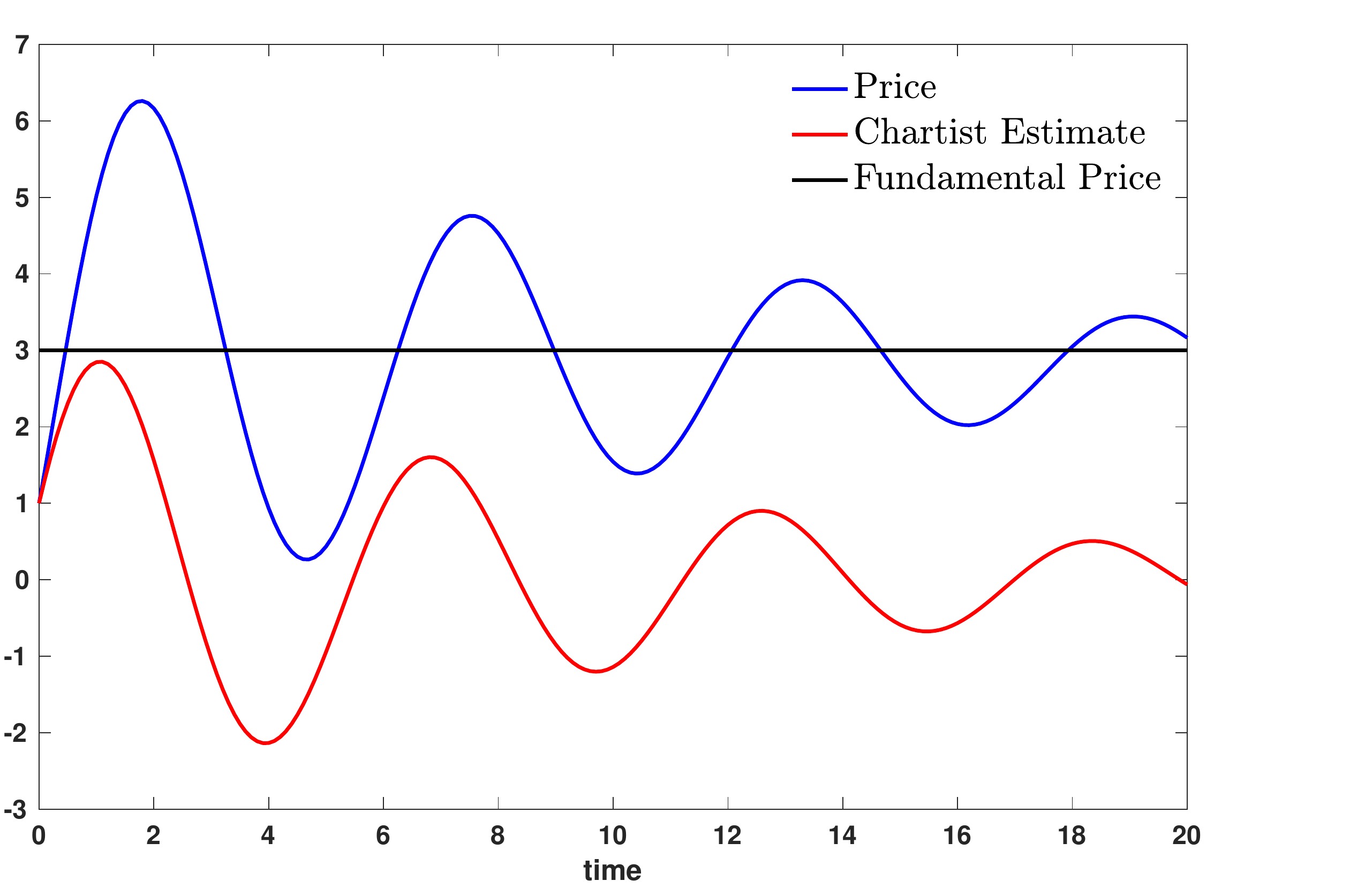}
\hfill
\includegraphics[width=0.45\textwidth]{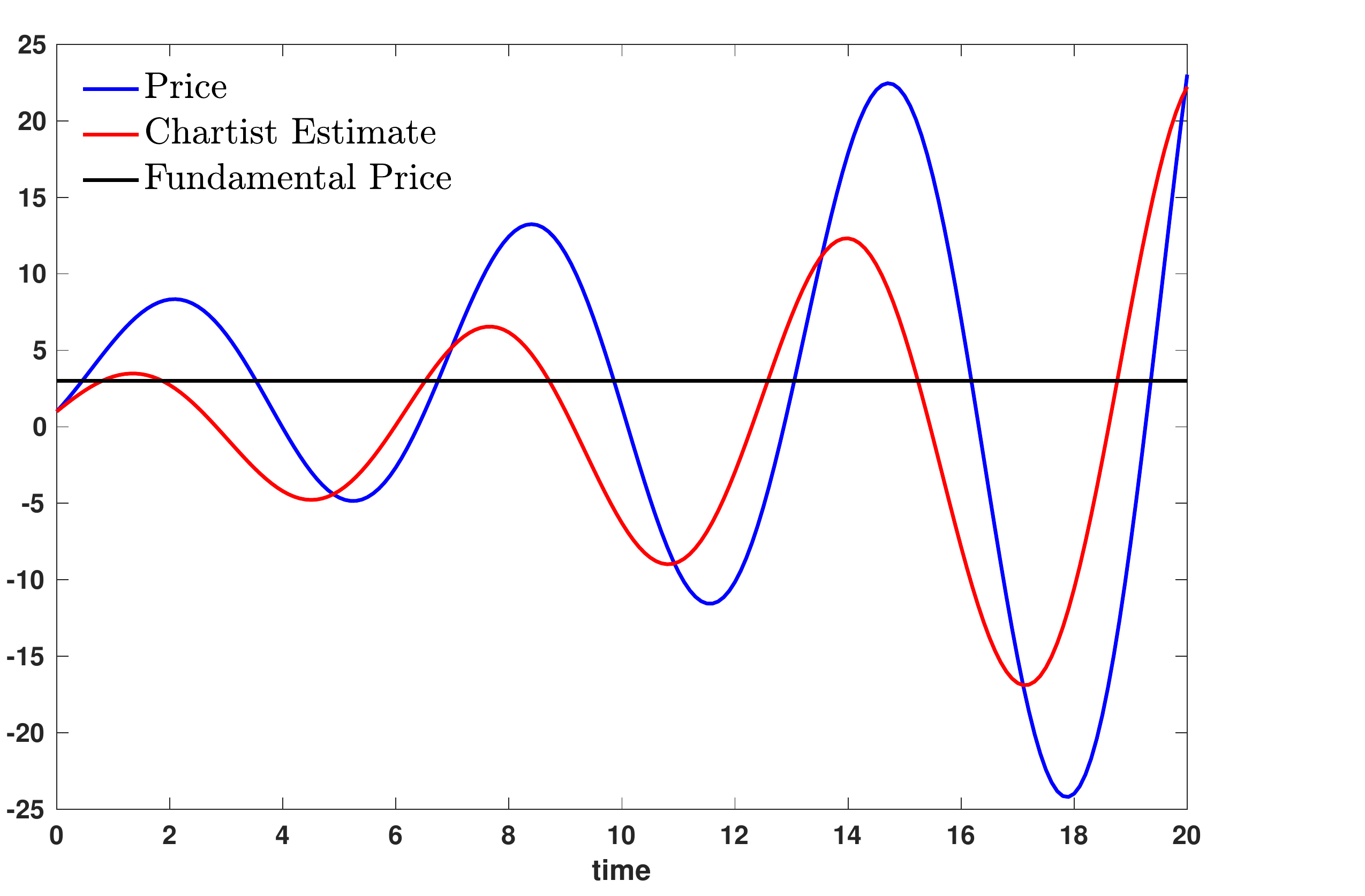}
\end{center}
\caption{LHS: Stability Region III with parameters $a=1.2,\ b=2,\ \gamma=1,\ \epsilon=1,\ r=0.1$.
RHS: Stability Region IV with parameters $a=1,\ b=2.2,\ \gamma=1,\ \epsilon=1,\ r=0.1$.}\label{OscDyn}
\end{figure}

\begin{figure}[h!]
\begin{center}
\includegraphics[width=0.45\textwidth]{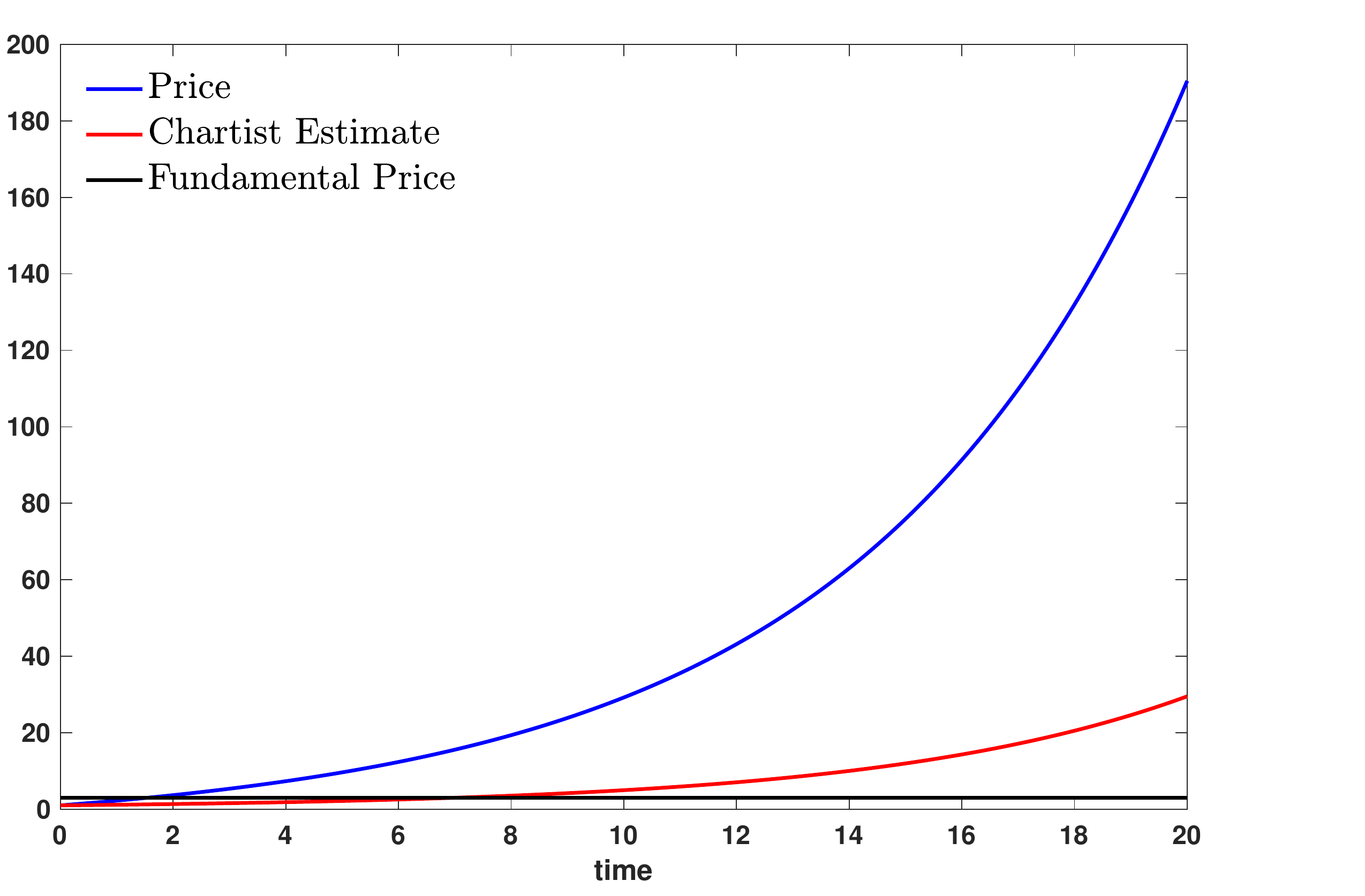}
\hfill
\includegraphics[width=0.45\textwidth]{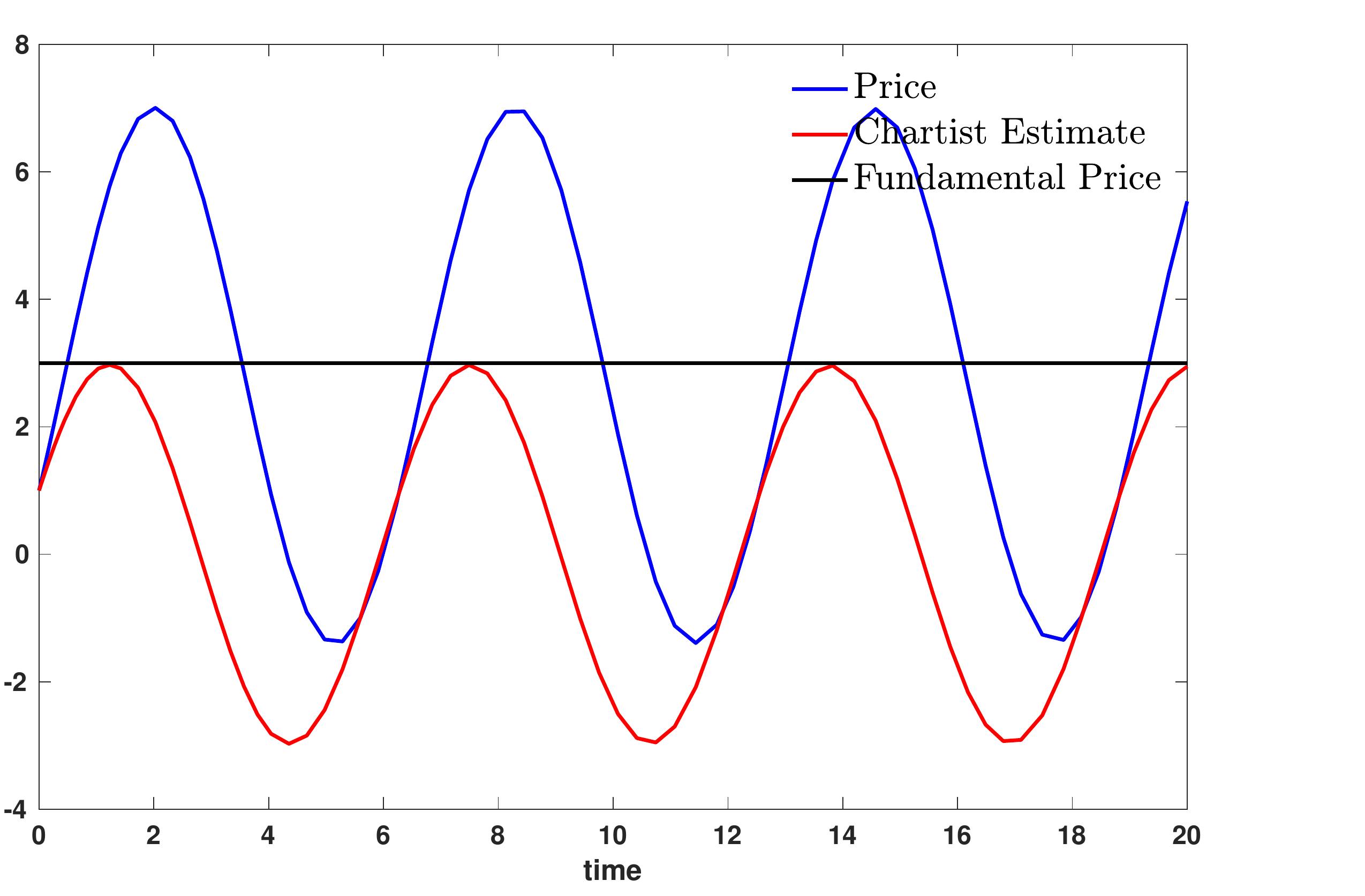}
\end{center}
\caption{LHS:Stability Region V with parameters $a=0.01,\ b=1.25,\ \gamma=1,\ \epsilon=1,\ r=0.1$.
RHS: Stability border $b=\epsilon +\gamma\ a$ with parameters $a=1,\ b=2,\ \gamma=1,\ \epsilon=1,\ r=0.1$.}\label{BUDyn}
\end{figure}

\clearpage

\section{Asymptotic Limits}
In this section we study the asymptotic limits of the Beja-Goldman model \eqref{BG}.
First, we define the precise form of a singular perturbation problem and then introduce the main tools in order to study asymptotic limits. 

\begin{definition}
Given an open set $D\subset \R^m$ and an analytic function $\xi: D\to\R^m$. Then the set
$$
\mathcal{V}(\xi):=\{  x\in\R^m, \xi(x)=0\}
$$
is called zero set of $\xi$. 
\end{definition}

\begin{definition}
We call a parameter dependent system of autonomous ordinary differential equations of the form
\begin{align}\label{sing}
\dot{x} = h(x,\epsilon)= h^{(0)}(x) + \epsilon h^{(1)}(x)+...,
\end{align}
with analytical function $h: D\times [0,\epsilon_0)\to R^m,\ \epsilon_0>0, D\subseteq \R^m$ a singular perturbation problem, 
provided that the zero set of  $h^{(0)}$ does not solely consists of isolated points. 
\end{definition}
By rescaling the time $\tau=\epsilon\ t$ we can rewrite \eqref{sing} as:
\begin{align}\label{rescaling}
\dot{x} = \frac{1}{\epsilon}h(x,\epsilon)= \epsilon^{-1} h^{(0)}(x) +  h^{(1)}(x)+...\quad .
\end{align}
The goal is to determine the limit system as $\epsilon \to 0$. %study the limit $\epsilon \to 0$ and to determine the corresponding limit system. 
By analyzing the limit system, which is of lower dimension and thus in general easier to analyze, one expects to gain information on
the original system \eqref{sing} for small $\epsilon$. In particular, the limit system is (for small $\epsilon$) a good approximation to the original system. Examples of such asymptotic reductions are given by \cite{grad1963asymptotic, frank2018quasi, heineken1967mathematical}.
To our knowledge such asymptotic limits have not been applied rigorously to any economic model. Formally (i.e. without any proof of convergence), singular perturbation models can be
studied by an asymptotic expansion also known as Hilbert expansion \cite{mckean1967chapman}. The ansatz for the solution $x$ of the system \eqref{sing} reads:
\begin{align}\label{AE}
x = x_0 + \epsilon\ x_1 + \epsilon^2\ x_2+ ...
\end{align}
As a next step one inserts the  asymptotic expansion \eqref{AE} into the original dynamical system. 
Then one can deduce the reduced system by comparing different orders of $\epsilon$. As one would expect this methodology does not always lead to the correct reduction as no convergence result is provided. Therefore, in this paper we consider a rigorous reduction approach based on the early works of Tikhonov and Fenichel \cite{fenichel1979geometric, tikhonov1952systems} as well as Hoppensteadt \cite{Hoppensteadt}. More precisely, we follow the constructive appraoch introduced by Goeke, Noethen and Walcher \cite{goeke2014constructive, noethen2011tikhonov} and use Lax, Selinger, Walcher \cite{LaxHoppensteadt} for an optimized convergence result. We will employ the following theorem as presented in \cite{goeke2014constructive}. 

\begin{theorem} \label{theo}
Given a system of the form \eqref{rescaling} for analytical $h$. We assume there exist a point $x_0$ in the zero set $\mathcal{V}(h^{(0)})$ of $h^{(0)}$, such that 
$rank(Dh^{(0)}(x_0))=r$ is maximal in a neighborhood of $x_0$. Then there exists a neighborhood $U$ of $x_0$ such that $\mathcal{U}= U\cap \mathcal{V}(h^{(0)})$
is a $s-$dimensional submanifold. If a sum decomposition of the form
$$
\R^m = kern\ Dh^{(0)}(x_0) \oplus im\ Dh^{(0)}(x_0),
$$
exists, then the following results hold:
\begin{itemize}
\item There exist a product decomposition
$$
h^{(0)}(x) = K(x)\ \mu(x),\quad x\in U,
$$
with analytic functions
$$
K: \R^m \to \R^{m \times r},\quad \mu: \R^m \to \R^r
$$
such that $rank(K(x_0))= rank(D\mu(x_0))=r$ holds. Furthermore, the zero set $Y$ of $\mu$ satisfies $Y\cap U = \mathcal{U}$. 
\item Define
$$
Q(x):= Id-K(x) (D\mu(x)K(x))^{-1}\ D\mu(x).
$$
Then the system 
\begin{align}\label{redSys}
x'=Q(x)\ h^{(1)}(x)
\end{align} is well defined on $U$. 
\item If there exists a $\mu>0, $ such that
$$
Re(\sigma^*(Dh^{(0)}(x_0))) \leq -\mu
$$
is fulfilled (i.e. the real part of every nonzero eigenvalue of the Jacobian $Dh^{(0)}$ at point $x_0$ is uniformly bounded from above by a negative constant), then there exists a $T>0$ and a neighborhood $U^*\subset U$ of $\mathcal{U}$ such that solutions of \eqref{rescaling} with initial conditions in $U^*$ for all $t_0>0$
on $[t_0,T]$ uniformly converge to the solution of the reduced system \eqref{redSys} on $\mathcal{U}$ as $\epsilon \to 0$. 
\end{itemize}
\end{theorem}

\begin{remark}
\begin{itemize}
\item $\mathcal U$ is often called the slow manifold, which we also will do in the following.
\item One should emphasize that the convergence result stated above is only valid on an finite interval. After leaving the Interval (i.e. $\tau >T$) the solution may start to oscilate, blow-up or show any other behaviour; see e.g. \cite{kruff2019rosenzweig} Section 4.1.
\item In \cite{LaxHoppensteadt}, conditions are presented such that the convergence holds - loosely spoken - for all positive times. In particular, Proposition 2.10 therein implies that the existence of exactly one stationary point $z$ which is exponentially stable (i.e. the linearization of the reduces system at $z$ has only one negative eigenvalue) in tie with a technical condition is sufficient for the desired convergence for all positive times.
%\item A solution of \eqref{rescaling} with initial condition $x_0\in \R^m$ converges to the solution of \eqref{redSys} with initial conditions $\tilde x_0$, where $\tilde x_0$ is the solution of
%	\begin{align*}
%	 \mu(x_0)&=\mu(\tilde x_0),\quad \tilde x_0\in \mathcal U 
%	\end{align*}
\end{itemize}
\end{remark}

In the rest of the chapter we show that the Beja-Goldman model is a singular perturbation problem in the liquid market limit and liquid chartist limit. Then we derive the corresponding reductions by use of Theorem \ref{theo}. Finally, we analyze the reduced systems and present several numerical tests.

%This theorem considers a system in Tichonov normal form.
%\begin{definition}
%We call a system of ordinary differential equations in Tikhonov normal form if is satisfies the following structure.
%\begin{subequations}
%\begin{align*}
%&\dot{y}_1 = f(\tau, y_1,y_2)= f^{(1)}(\tau, y_1,y_2)+ \epsilon f^{(2)}(\tau, y_1,y_2)+...\\
%&\dot{y}_2 = \frac{1}{\epsilon} g(\tau, y_1,y_2)=\frac{1}{\epsilon} g^{(0)}(\tau, y_1,y_2) + \epsilon g^{(1)}(\tau, y_1,y_2)+... 
%\end{align*}
%\end{subequations}
%The functions $f,g$ are given by:
%\begin{align*}
%& f:  [0,\infty ) \times D\times [0,\epsilon_0) \to \R^s,\\
%& g: [0,\infty ) \times D\times [0,\epsilon_0)\to  \R^r,
%\end{align*}
%with $D=D_1\times D_2\subseteq \R^m,\ D_1\subset\R^s,\ D_2\subset\R^r$.
%\end{definition}

\subsection{Liquid Market Limit}
In this section we study the liquid market limit $\epsilon\to 0$ which corresponds to an infinite market depth. 
We treat the parameter $\gamma$ as a given parameter.

\begin{proposition}
The Beja-Goldman model as defined in \eqref{BG} is a singular perturbation model in the liquid market regime. 
\end{proposition}
\begin{proof}
The slow manifold is given by the zero set of:
\begin{align*}
&h^{(0)}(P,\Psi):=  \begin{pmatrix}  a\ (F-P) +b\ (\Psi-r)\\  \frac{a}{\gamma}\ (F-P) +\frac{b}{\gamma}\ (\Psi-r) \end{pmatrix}. 
\end{align*}
Thus we have $\mathcal{V}(h^{(0)}) := \{(P,\Psi)^T\in\R^2: P= \frac{b}{a} (\Psi-r) + F  \}$, which does not solely consist of isolated points. 
\end{proof}
As a next step we derive the reduced system. \

\begin{proposition}\label{PropRedLM}
Given the Beja-Goldman system as defined in \eqref{BG} and that $a \gamma> b$ holds, then in the liquid market limit the reduced system of the model \eqref{BG} reads
\begin{align}\label{RedModLM}
\Psi'= -\frac{a}{a\gamma-b}\ \Psi, 
\end{align}
where the price is given by $P= \frac{b}{a}(\Psi-r)+F$ and the convergence holds for all $\tau >0$. The unique solution of the reduced system is given by
$$
\Psi(\tau)=\Psi (0)\ \exp\left(-\frac{a}{a\gamma-b}\ \tau \right)
$$
and converges to zero, as $\tau \to \infty$.
\end{proposition}
\begin{proof}
We utilize Theorem \ref{theo} and define:
\begin{align*}
& h^{(1)}(P,\Psi):= \begin{pmatrix}0 \\ -\frac{\Psi}{\gamma} \end{pmatrix}.
\end{align*}
The Jacobian of $h^{(0)}$ is given by:
$$
Dh^{(0)}(P,\Psi)= \begin{pmatrix}  -a & b \\ -\frac{a}{\gamma} & \frac{b}{\gamma}  \end{pmatrix}
$$
and hence the conditions of Theorem \ref{theo} are satisfies as the nonzero eigenvalue $\frac{b-a\gamma}{\gamma}$ is simple and negative for $a \gamma> b$.
Next, we define
$$
K:= \begin{pmatrix}  1\\ \frac{1}{\gamma}  \end{pmatrix},\quad \mu(P,\Psi):= a(F-P)+b\ (\Psi-r).
$$
In particular, $h^0=K\ \mu$ holds and the Jacobian of $\mu$ reads:
$$
D\mu:=D\mu(P,\Psi)=(-a , \  b). 
$$
Then $D\mu K=-a+\frac{b}{\gamma}$ and we can define $Q$ as follows:
\begin{align*}
Q:&= I-K (D\mu\ K)^{-1}\ D\mu\\
&= \begin{pmatrix} 1 & 0\\ 0 & 1   \end{pmatrix}+ \frac{1}{a-\frac{b}{\gamma}} \begin{pmatrix}-a & b \\-\frac{a}{\gamma} & \frac{b}{\gamma}  \end{pmatrix}\\
& =  \frac{1}{a-\frac{b}{\gamma}}\begin{pmatrix} -\frac{b}{\gamma}& b \\ -\frac{a}{\gamma} & a    \end{pmatrix}.
\end{align*}
Hence, we get
\begin{align*}
\frac{d}{dt}\boldsymbol{X} = Q\ h^{(1)}(P,\Psi)= -\frac{1}{\gamma\ a-b}\begin{pmatrix} b \\ a \end{pmatrix}\ \Psi
\end{align*}
Then the reduced system reads
$$
\Psi' = -\frac{a}{\gamma a-b} \Psi,
$$
where $P$ is given by 
$$
P= \frac{b}{a}(\Psi-r)+F,
$$
or alternatively by the previously computed ODE. The convergence for all $\tau >0$ follows by Proposition 2.10 of \cite{LaxHoppensteadt}, using that there exists only one stationary point of the limit system, which is exponentially stable. (We omit a discussion of the additional technical condition that needs to be satisfied. But it can easily be proven that the condition holds using the stability properties of the original system \eqref{BG}).\\
Computing the explicit solution is then basic analysis.
\end{proof}

\begin{corollary}
The rigorously reduced system in Proposition \ref{PropRedLM}  is identical to the heuristically derived system with an asymptotic expansion of type \eqref{AE}. 
\end{corollary}

\begin{proof}
We insert a Hilbert expansion:
\begin{align*}
&P(t)= P_0 + \epsilon\ P_1+\epsilon^2 \ P_2+...\\
&\Psi(t)= \Psi_0 + \epsilon\ \Psi_1+\epsilon^2 \ \Psi_2+...
\end{align*}
Then we can insert the previous expansion in our model and obtain for different orders:
\begin{align*}
&\text{Price equation:}\ &&\mathcal{O}(\frac{1}{\epsilon}): \quad P_0 = \frac{b}{a} (\Psi_0-r)+F\\
&\quad && \mathcal{O}(0): \quad \dot{P}_0 = - a\ P_1 +b\ \Psi_1\\
&\text{Chartist equation:}\ &&\mathcal{O}(\frac{1}{\epsilon}): \quad  P_0 = \frac{b}{a} (\Psi_0-r)+F\\
&\quad && \mathcal{O}(0): \quad \dot{\Psi}_0 =\frac{1}{\gamma} \left[      - a\ P_1 + b\ \Psi_1 -\Psi_0     \right] \\
\end{align*}
We use the price equation of order $\mathcal{O}(0)$ to rewrite the chartist equation of order $\mathcal{O}(0)$ as follows:
\begin{align*}
\dot{\Psi}_0 &=\frac{1}{\gamma} \left[      - a\ P_1 + b\ \Psi_1 -\Psi_0     \right]\\
&= \frac{1}{\gamma} \left[     \dot{P}_0 -\Psi_0     \right]
\end{align*}
As next step we utilize the equality
$$
\dot{P}_0=\frac{b}{a} \dot{\Psi}_0,
$$
which is obtained by differentiating the price equation of order $\mathcal{O}(\frac{1}{\epsilon})$. We get
\begin{align*}
\dot{\Psi}_0 &=\frac{1}{\gamma} \left[     \frac{b}{a} \dot{\Psi}_0 -\Psi_0     \right]\\
\iff \dot{\Psi}_0 &=-\frac{a}{a\gamma-b}\ \Psi_0,
\end{align*}
and thus the statement is shown. 
\end{proof}

\begin{itemize}
\item The solutions of the limit system \eqref{RedModLM} satisfies
\begin{align*}
&\lim\limits_{\tau\to\infty} \Psi(\tau) = 0,\\
&\lim\limits_{\tau\to\infty} P(\tau) = F-r \frac{b}{a}.
\end{align*}
The result is not surprising as the condition $a \gamma> b$ implies $a>\frac{1}{\gamma}(b-\epsilon)$ and thus corresponds to the original system \eqref{BG} being stable.
%\item As described in Remark ref, one can also descibe the convergence properties for a single solution: The solution $\Phi(\tau, \epsilon)$ of the originial system \eqref{BG} with repespect to the initial values $(P_0, \Psi_0)$ converges (as $\epsilon\to 0$) to the solution $\widetilde \Phi(\tau)$ of the limit system \eqref{RedModLM} with respect to the inital values
%	\[
%	 (\tilde P_0, \tilde \Psi_0):= 
%	\]
\item If $a \gamma< b$ holds then the slow manifold still exists and is still invariant, but now is repelling, i.e. every solution not starting on the slow manifold will never reach the slow manifold. Solutions starting on the slow manifold will now diverge to $\pm\infty$ as $\tau \to \infty$, which corresponds to the unstable case for system \eqref{BG}.
\item The case $a\gamma =b$ is degenerate: There exists no slow manifold in this case. The dynamic of the original system for different small values of $\epsilon$ is depicted in Figures \ref{DegMarketPos} and \ref{DegMarketNeg}. 
\end{itemize}

\begin{figure}[h!]
\begin{center}
\includegraphics[width=0.45\textwidth]{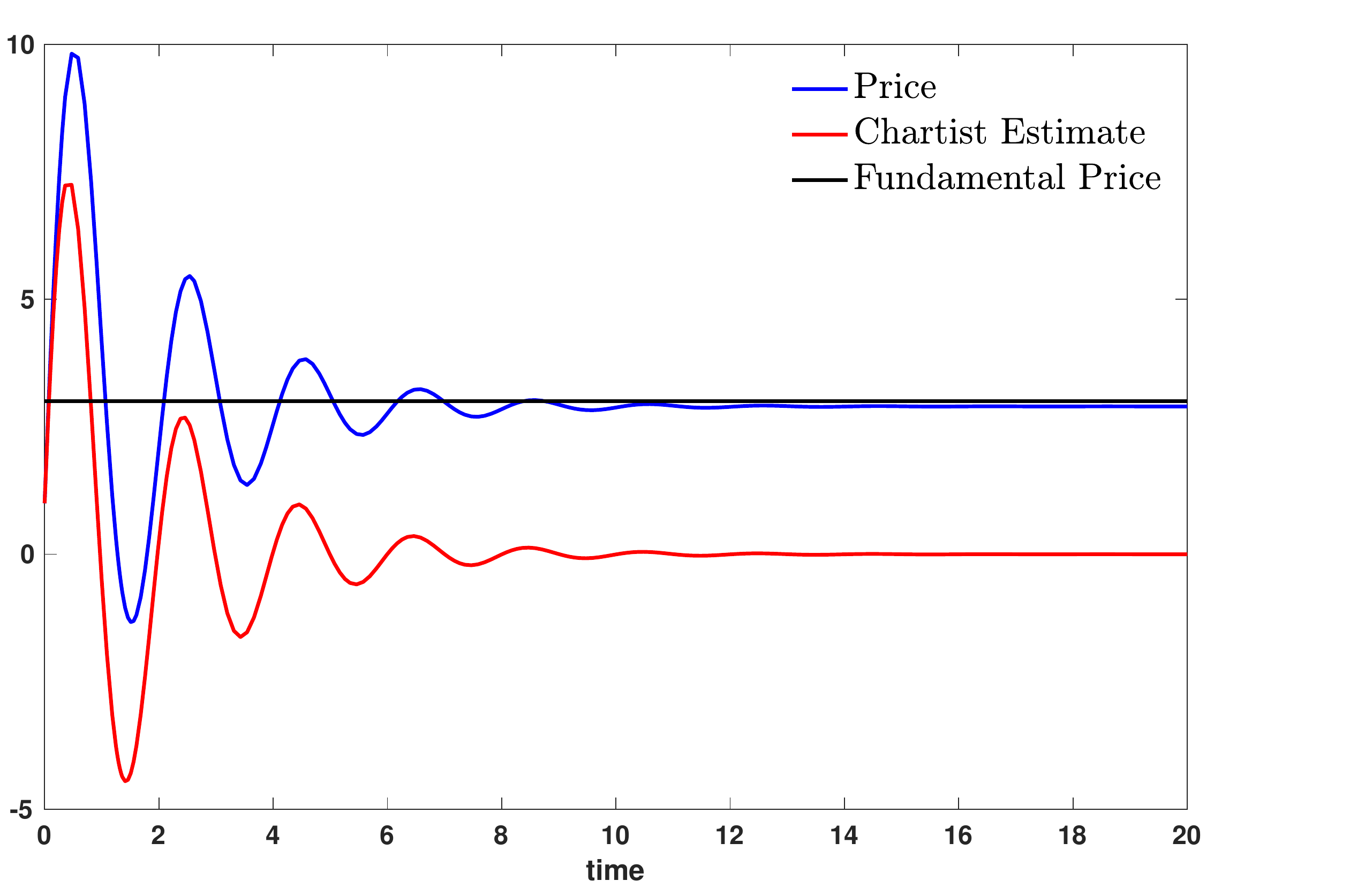}
\hfill
\includegraphics[width=0.45\textwidth]{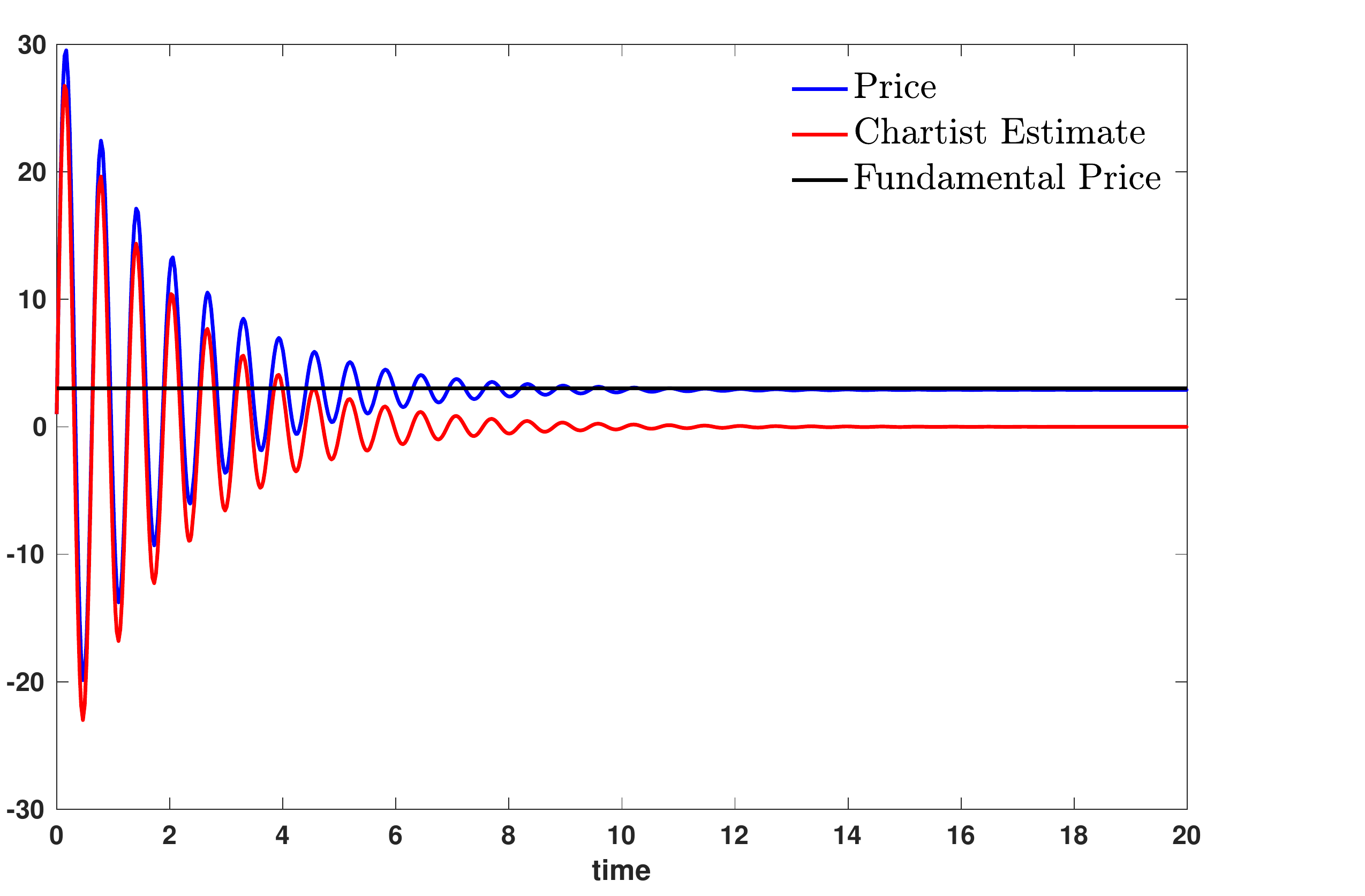}
\end{center}
\caption{LHS: Degenerate case with parameters $a=1,\ b=1,\ \gamma=1,\ \epsilon=0.1,\ r=0.1$.
RHS: Degenerate case with parameters $a=1,\ b=1,\ \gamma=1,\ \epsilon=0.01,\ r=0.1$.}\label{DegMarketPos}
\end{figure}

\begin{figure}[h!]
\begin{center}
\includegraphics[width=0.45\textwidth]{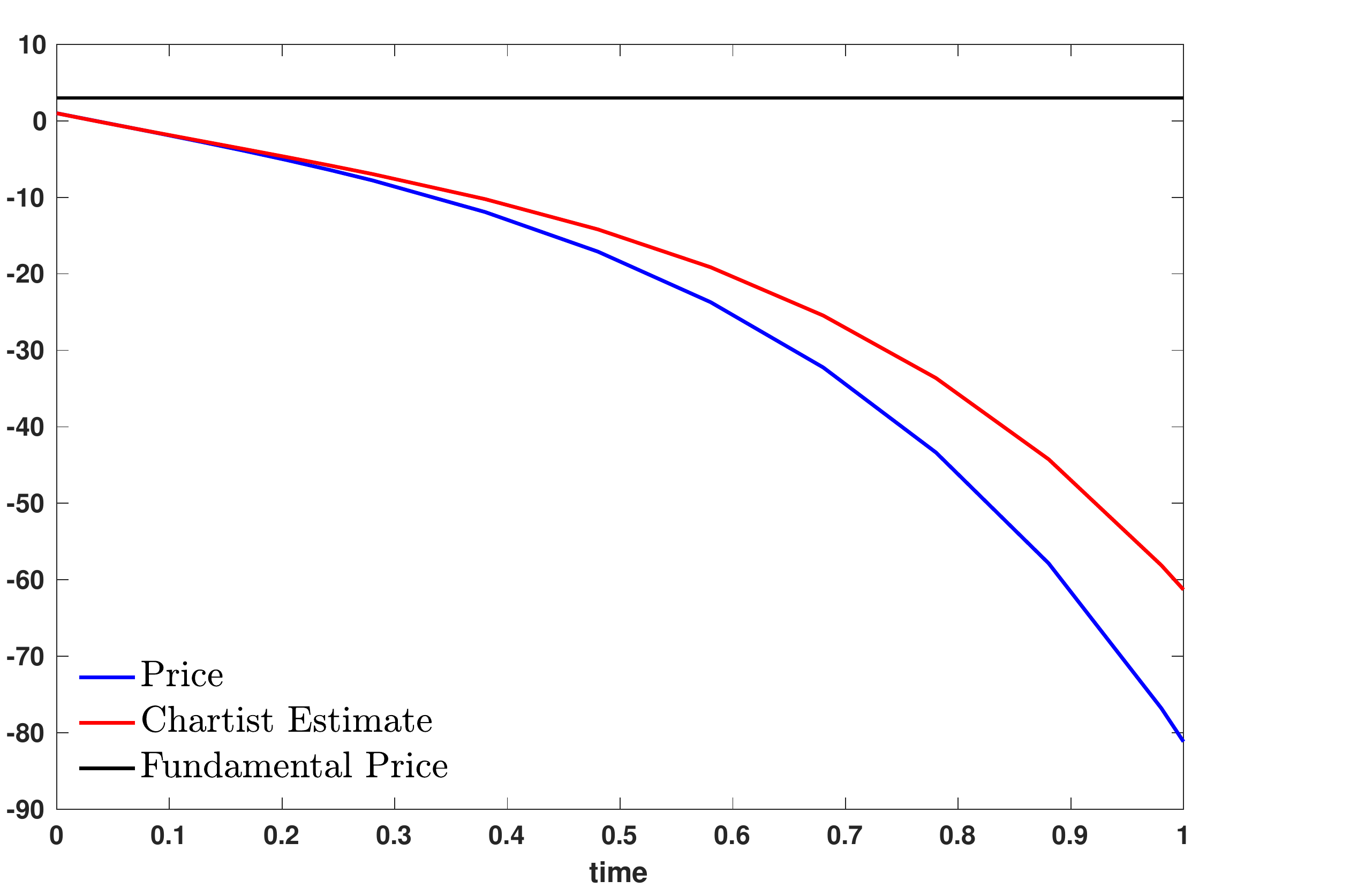}
\hfill
\includegraphics[width=0.45\textwidth]{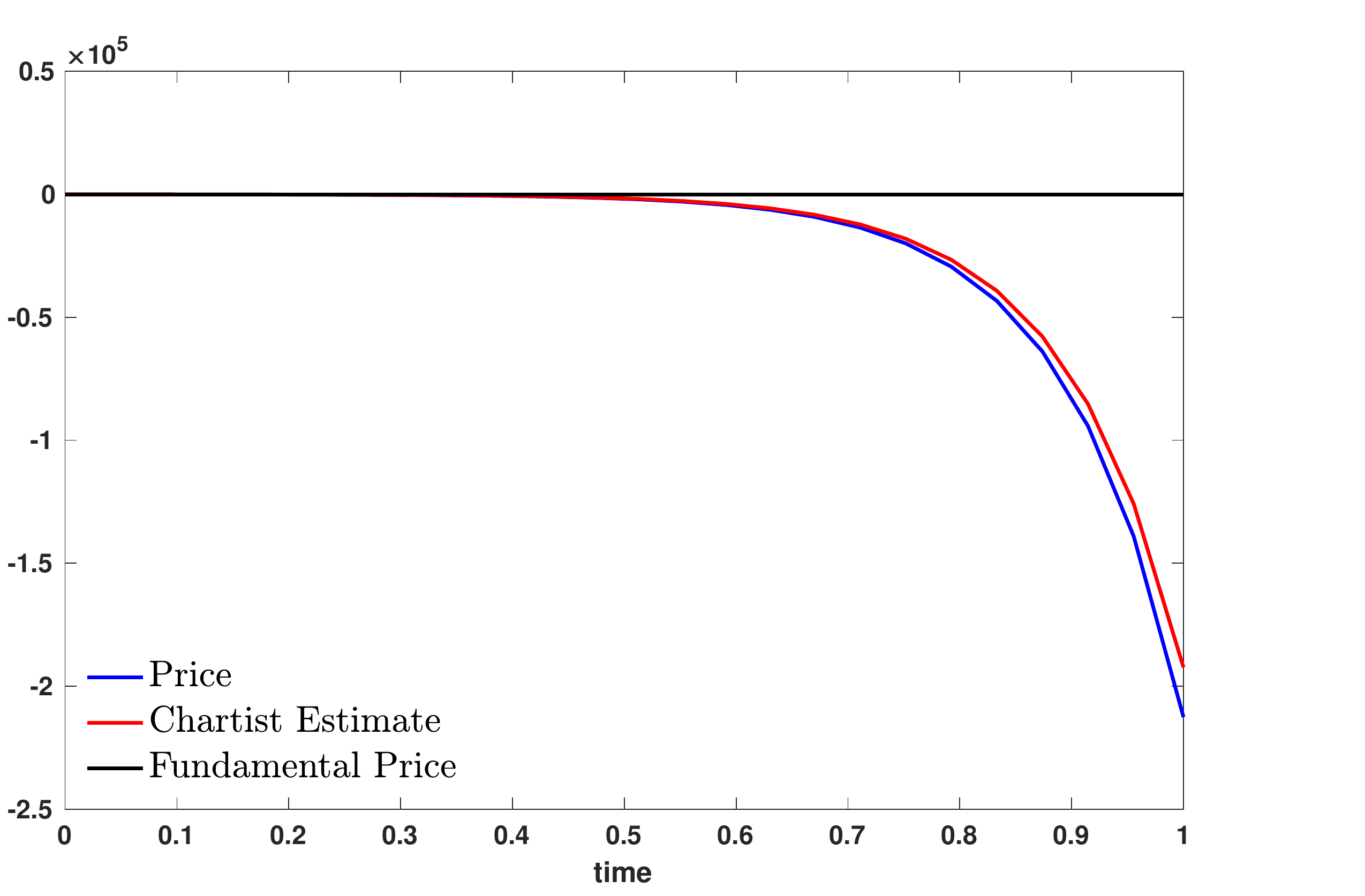}
\end{center}
\caption{LHS: Degenerate case with parameters $a=1,\ b=1,\ \gamma=1,\ \epsilon=-0.1,\ r=0.1$.
RHS: Degenerate case with parameters $a=1,\ b=1,\ \gamma=1,\ \epsilon=-0.01,\ r=0.1$.}\label{DegMarketNeg}
\end{figure}

\clearpage

%\textcolor{red}{Welche Anfangswerte hast du für die Numerik verwendet? Für Anfangswerte $(P_0,\Psi_0)$ vom Originalsystem ergeben sich die richtigen reduzierten Anfangswerte $(\tilde P_0, \tilde \Psi_0)$
%\begin{align*}
% \tilde P_0 &= P_0 - \gamma (\Psi_0-\tilde \Psi_0)\\
% \tilde \Psi_0 &= \frac{1}{a\gamma-b}(a(F-P_0)+a\gamma\Psi_0-br)
%\end{align*}
%Dies wirkt vlt etwas unintuitiv, ergibt sich aber aus der Slow Manifold Bedingung $a(F-\tilde P_0)+b(\tilde \Psi_0-r)=0$ in Verbindung mit dem 1. Integral $P-\gamma \Psi$ des schnellen Systems $x'=h^0(x)$ (liefert $P_0-\gamma \Psi_0=\tilde P_0-\gamma \tilde\Psi_0$).
%}
Finally, we present some numerical results. We have computed the solutions of the original Beja-Goldman model with the Matlab solver \texttt{ode15s} which is especially designed for stiff ordinary differential equations. From Figure \ref{RedM} we can deduce that the original model asymptotically converges to the reduced model if $a \gamma> b$. This asymptotic behavior of both models for different parameters $\epsilon$ can be seen in Figure \ref{ErrorM}. Thus, our numerical test validate our analysis that in this case the Beja-Goldman model converges to the reduced model \eqref{RedModLM} as $\epsilon\to 0$. 
In addition we study the case if no reduction exists ($a\gamma <b$). In fact the slow manifold still exists but is repelling. In Figure \ref{RepMani} we show the cases of initial values laying on the slow manifold or not laying on the slow manifold. One can observe that the original system is always repelling but diverges slower in the case of initial values on the slow manifold.

\begin{figure}[h!]
\begin{center}
\includegraphics[width=0.6\textwidth]{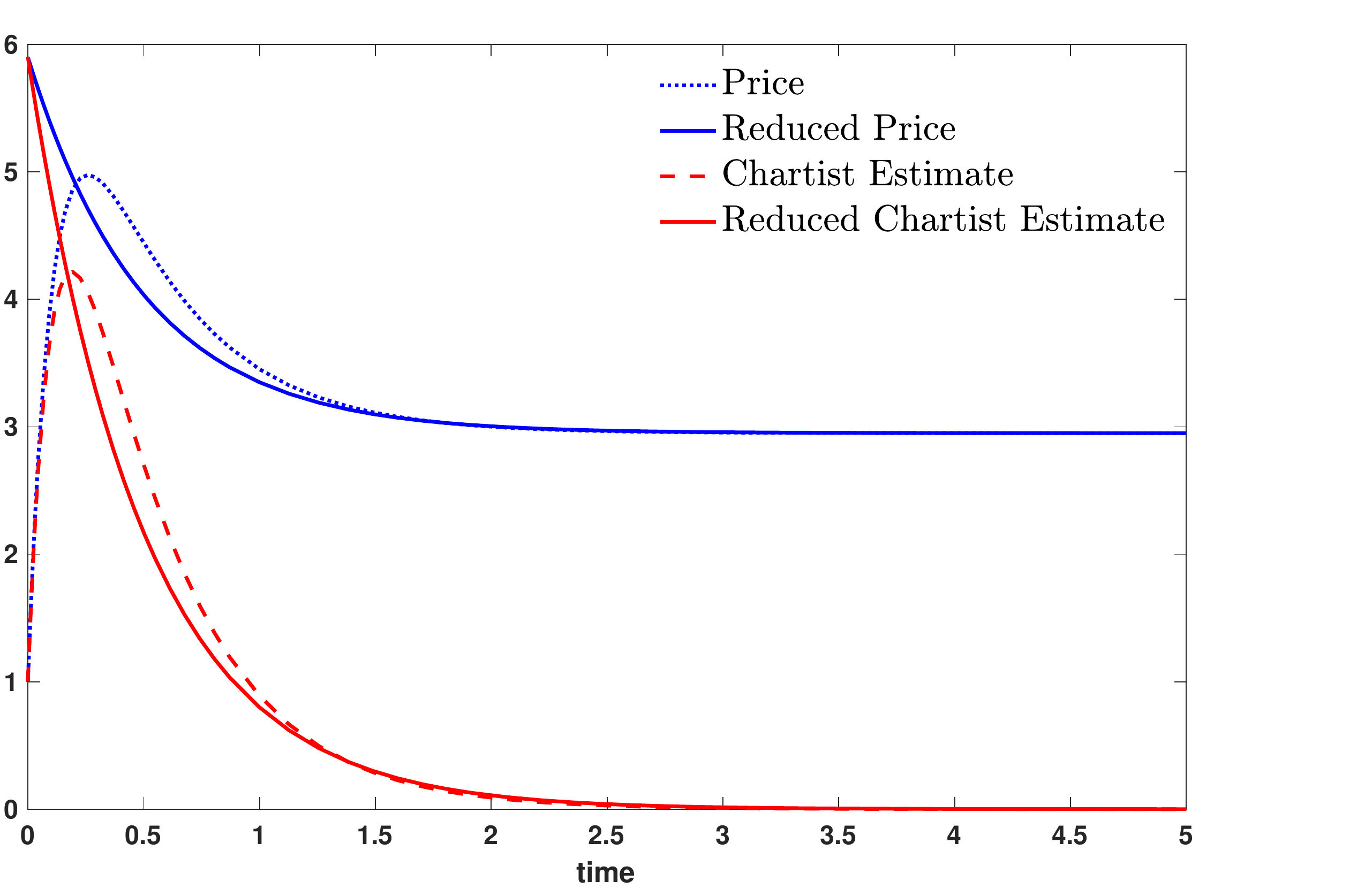}
\end{center}
\caption{Time evolution of the original Beja Goldman model and reduced model in the liquid market regime ($\epsilon= 0.1, a=2, b=1,\gamma=1,F=3,r=0.1$).}\label{RedM}
\end{figure}

\begin{figure}[h!]
\begin{center}
\includegraphics[width=0.45\textwidth]{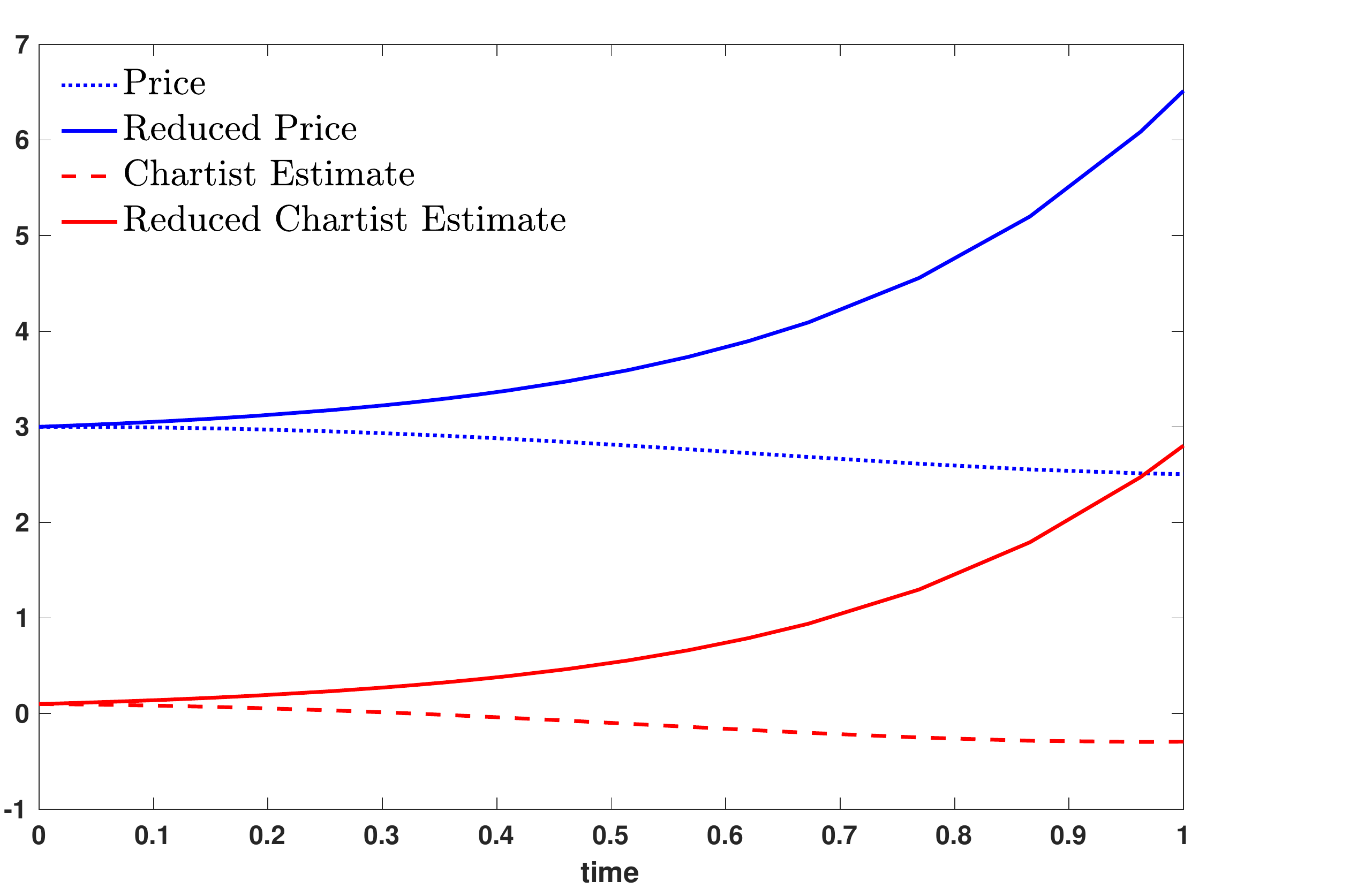} \hfill
\includegraphics[width=0.45\textwidth]{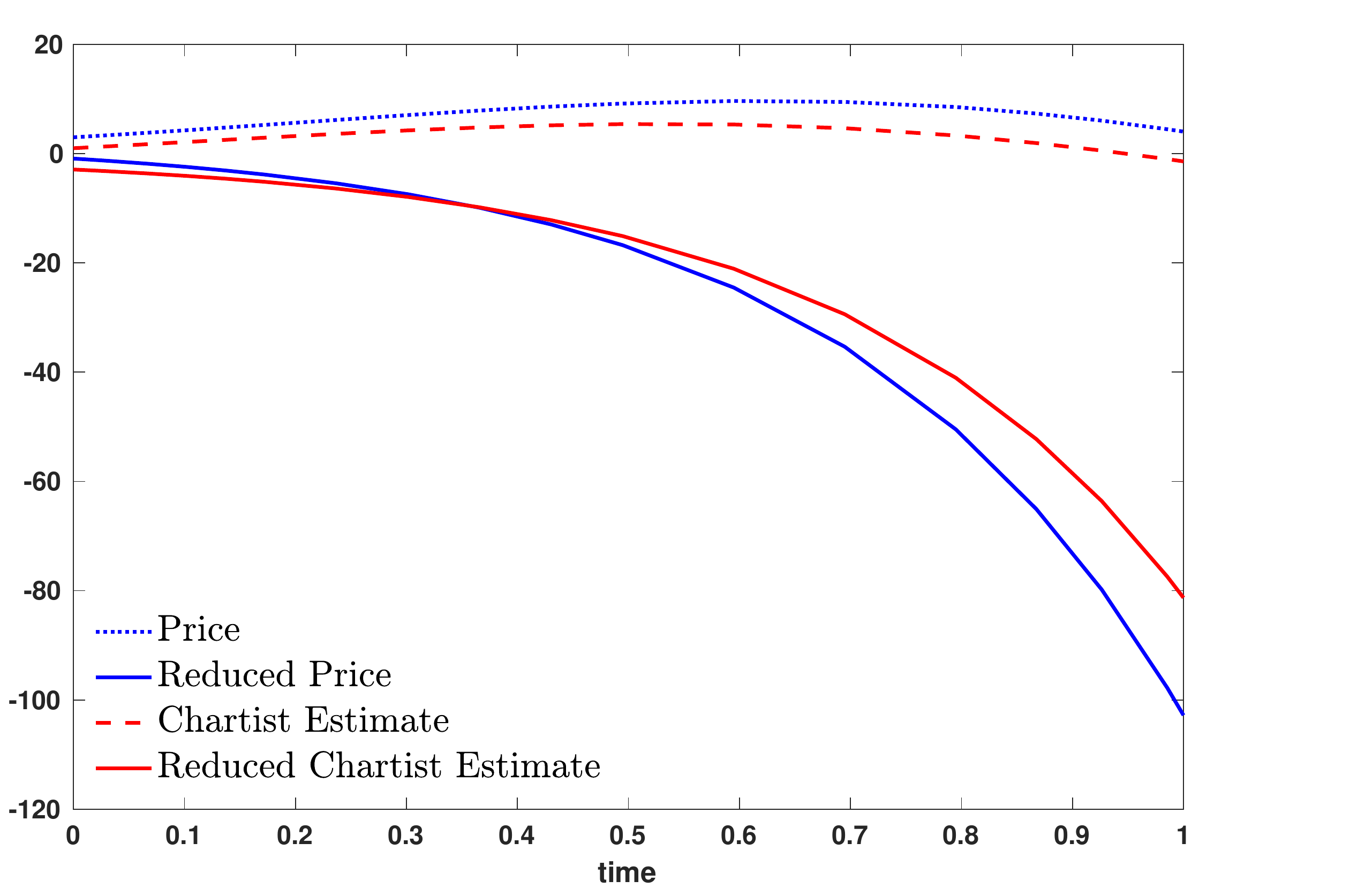}
\end{center}
\caption{Non-stable case with diverging reduced model. Parameters are given by $\epsilon= 0.1, a=1, b=1.3,\gamma=1,F=3,r=0.1$.
LHS: Initial values $P_0=3,\ \Psi = 0.1$ are located on the slow manifold.
RHS: Initial values $P_0=3,\ \Psi = 1$ are not located on the slow manifold.
 }\label{RepMani}
\end{figure}

\begin{figure}[h!]
\begin{center}
\includegraphics[width=0.6\textwidth]{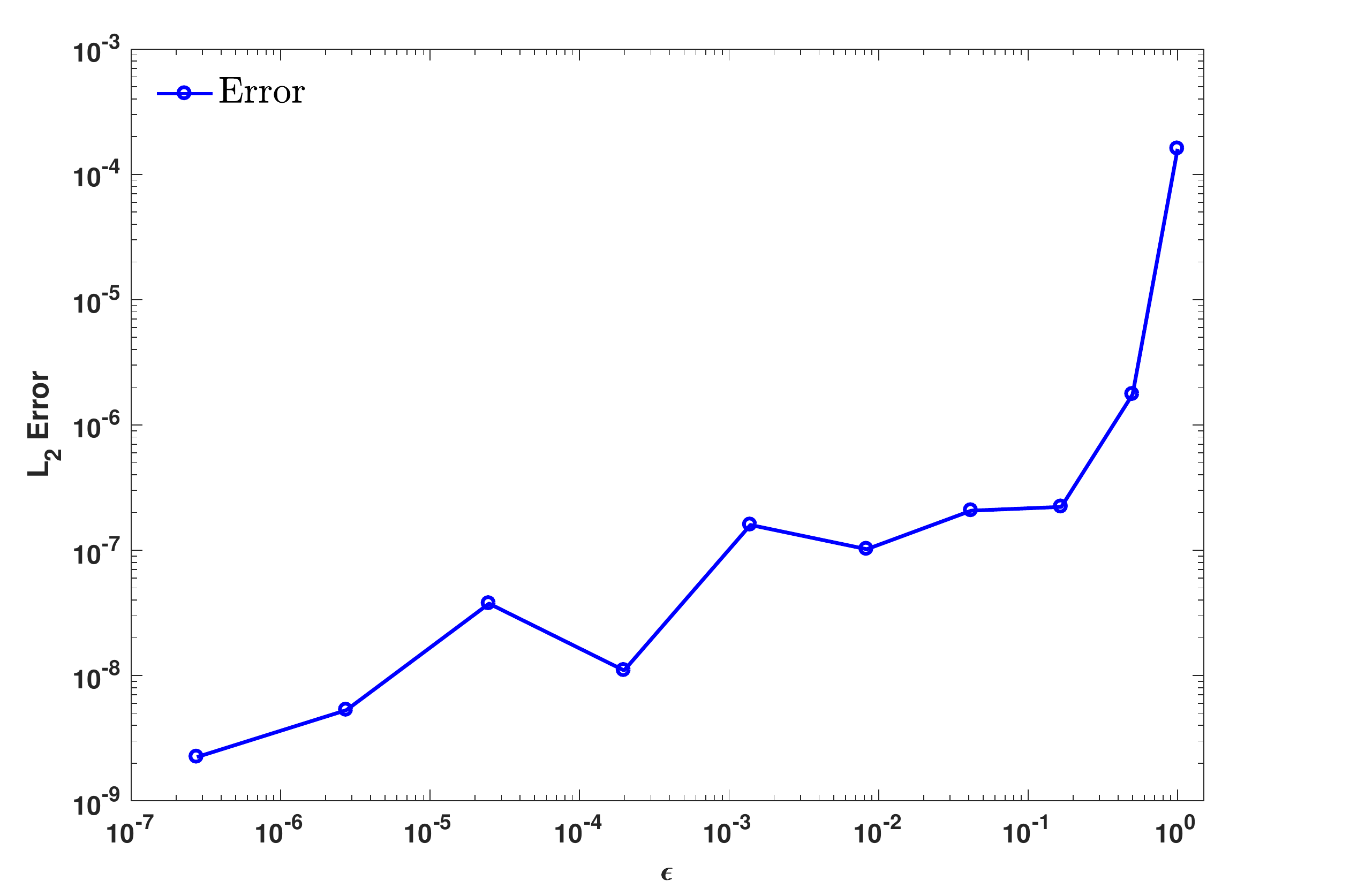}
\end{center}
\caption{The $L_2$ distance at time $t=10$ of the original Beja-Goldman model and the reduced model for different values $\epsilon$. The other parameters are chosen to be: $ a=2, b=1,\gamma=1,F=3,r=0.1$ }\label{ErrorM}
\end{figure}

\clearpage
\subsection{Liquid Chartist Limit}
In this section we study the limit $\gamma \to 0$ which corresponds to a infinite fast reaction speed of chartists.
The parameter $\epsilon$ is treated as a constant and is fixed. 
 
\begin{proposition}
The Beja-Goldman model as defined in  \eqref{BG} is a singular perturbation model in the liquid chartist regime. 
\end{proposition}
\begin{proof}
The slow manifold is given by the zero set of:
\begin{align*}
h^{(0)}(P, \Psi)=\begin{pmatrix} 0 \\  \frac{1}{\epsilon} \left( a[ F-P]+ b[\Psi -r]  \right)-\Psi\end{pmatrix}.
\end{align*}
Then the zero set $\mathcal{V}(h^{(0)}) := \{(P,\Psi)^T\in\R^2: \Psi= \frac{1}{\epsilon-b}  (a[F-P]+b\ r)  \}$ does not solely consist of isolated points. 
\end{proof}

\begin{proposition}
Let the Beja-Goldman model be as defined in \eqref{BG} and assume $\epsilon> b$. Then the reduced system in the liquid chartist limit reads
\begin{align}
\dot{P}&= \frac{1}{\epsilon} [ a\ (F-P)+b\ (\Psi-r)] \label{RCL1}\\
 &=\frac{a(F-P)-rb}{\epsilon-b}\label{RCL2}
\end{align}
and the convergence holds for all $\tau >0$. $\Psi$ is given by $\Psi= \frac{1}{\epsilon-b} [a(F-P)-r\ b]$. 
The unique solution of the stock price equation is given by:
\begin{align*}
P(\tau)= \exp(-\tfrac{a}{\epsilon-b} \tau) \left(P(0)-(F-\tfrac{rb}{a})\right) + F-\tfrac{rb}{a}.
\end{align*}
\end{proposition}

\begin{proof}
The result for finite time follows directly by Tikhonov's Theorem \cite{tikhonov1952systems}. One also can use Theorem \ref{theo} with
\begin{align*}
&h^{(1)}(P,\Psi):= \begin{pmatrix}  \frac{1}{\epsilon} \left( a[ F-P]+ b[\Psi -r]  \right)   \\ 0  \end{pmatrix}
\end{align*}
as well as
$$
K:=\begin{pmatrix} 0 \\ 1 \end{pmatrix},\quad \mu(P,\Psi):= \frac{1}{\epsilon} \left( a[ F-P]+ b[\Psi -r]  \right)-\Psi.
$$
A straightforward computation shows 
\[
 Q=\begin{pmatrix}
  1 & 0 \\ \tfrac{a}{b-\epsilon} & 0
 \end{pmatrix}
\]
which gives \eqref{RCL1}, where $\Psi= \frac{1}{\epsilon-b} [a(F-P)-r b]$. Inserting $\Psi$ yields \eqref{RCL2}. The convergence on infinite intervals follows again by \cite{LaxHoppensteadt} Proposition 2.10 (in the same manner as in the proof of Proposition 4).  Solving the reduced equation is again a simple computation.
%Especially $h^{(0)}= K\ \mu$ holds and the Jacobian of $\mu$ reads:
%$$
%D\mu=(-\frac{a}{\epsilon}, \frac{b}{\epsilon}-1). 
%$$
%As before we compute
%\begin{align*}
%Q&= I- K(D\mu\ K)^{-1} D\mu\\
%&=\begin{pmatrix} 1 & 0 \\ 0 & 1 \end{pmatrix}- \frac{1}{\frac{b}{\epsilon}-1} \begin{pmatrix} 0 & 0 \\ -\frac{a}{\epsilon} & \frac{b}{\epsilon}-1  \end{pmatrix}
%\end{align*}
%Thus, the reduced system reads
%$$
%\dot{\boldsymbol{X}} = Q\ h^{(1)}= \begin{pmatrix}   \frac{1}{\epsilon} \left( a[ F-P]+ b[\Psi -r]  \right)\\  \frac{a}{b-\epsilon}\frac{1}{\epsilon} \left( a[ F-P]+ b[\Psi -r]  \right)   \end{pmatrix}
%$$
%The derived ODE for the quantity $\Psi$ can be equivalently replaced by the algebraic relation
%$$
%\Psi=  \frac{1}{\epsilon -b}  \left( a[ F-P]-r\ b  \right),
%$$
%which is the reduced system. 
\end{proof}

\begin{corollary}
The rigorously reduced system \eqref{RCL1} is identical to the heuristically derived system with an asymptotic expansion of type \eqref{AE} . 
\end{corollary}
\begin{proof}
As before, we insert a Hilbert expansion:
\begin{align*}
&P(t)= P_0 + \gamma\ P_1+\gamma^2 \ P_2+...\\
&\Psi(t)= \Psi_0 + \gamma\ \Psi_1+\gamma^2 \ \Psi_2+...
\end{align*}
Then comparing the  different orders, we get:
\begin{align*}
&\text{Price equation:}\ &&\mathcal{O}(0): \quad \dot{P}_0 = \frac{1}{\epsilon}[a(F-P_0)+b(\Psi_0-r)  ]      \\
&\text{Chartist equation:}\ &&\mathcal{O}(\frac{1}{\gamma}): \quad  \Psi_0 =\frac{1}{\epsilon}[a(F-P_0)+b(\Psi_0-r)  ].
\end{align*}
\end{proof}

\begin{itemize}
\item The solutions of the limit system \eqref{RCL1} satisfies again
\begin{align*}
&\lim\limits_{\tau\to\infty} \Psi(\tau) = 0,\\
&\lim\limits_{\tau\to\infty} P(\tau) = F-r \frac{b}{a}.
\end{align*}
The result is not surprising as the condition $\epsilon> b$ implies $a>\frac{1}{\gamma}(b-\epsilon)$ and thus corresponds to system \eqref{BG} being stable.
\item If $\epsilon< b$ holds then the slow manifold still exists and is still invariant, but now is repelling. Solutions starting on the slow manifold will now diverge to $\pm\infty$ as $\tau \to \infty$, which corresponds to the unstable case for system \eqref{BG}.
\item The case $\epsilon =b$ is degenerate: There exists no slow manifold in this case. As shown previously in the liquid market regime it is possible to study the original system for different small values $\gamma$. 
\end{itemize}

%\begin{figure}[h!]
%\begin{center}
%\includegraphics[width=0.45\textwidth]{figs/results/DegenerateEstimate/DegenerateEps=B01.pdf}
%\hfill
%\includegraphics[width=0.45\textwidth]{figs/results/DegenerateEstimate/DegenerateEps=B001.pdf}
%\end{center}
%\caption{LHS: Degenerate case with parameters $a=1,\ b=1,\ \gamma=0.1,\ \epsilon=1,\ r=0.1$.
%RHS: Degenerate case with parameters $a=1,\ b=1,\ \gamma=0.01,\ \epsilon=1,\ r=0.1$.}
%\end{figure}
%
%\begin{figure}[h!]
%\begin{center}
%\includegraphics[width=0.45\textwidth]{figs/results/DegenerateEstimate/DegenerateLC-01.pdf}
%\hfill
%\includegraphics[width=0.45\textwidth]{figs/results/DegenerateEstimate/DegenerateLC-001.pdf}
%\end{center}
%\caption{LHS: Degenerate case with parameters $a=1,\ b=1,\ \gamma=-0.1,\ \epsilon=1,\ r=0.1$.
%RHS: Degenerate case with parameters $a=1,\ b=1,\ \gamma=-0.01,\ \epsilon=1,\ r=0.1$.}
%\end{figure}

%\textcolor{red}{Welche Anfangswerte hast du für die Numerik verwendet? Für Anfangswerte $(P_0,\Psi_0)$ vom Originalsystem ergeben sich die richtigen reduzierten Anfangswerte $(\tilde P_0, \tilde \Psi_0)$
%\begin{align*}
% \tilde P_0 &= P_0 \\
% \tilde \Psi_0 &= \frac{1}{\epsilon-b}(a(F-P_0)-br)
%\end{align*}
%Wie vorher (hier nur intuitiver, weil Tikhonov Normalform) ergibt sich aus der Slow Manifold Bedingung \[\frac{1}{\epsilon}[a(F-\tilde P_0)+b(\tilde \Psi_0-r)]-\tilde \Psi_0=0\] in Verbindung mit dem 1. Integral $P$ des schnellen Systems $x'=h^0(x)$ (liefert $P_0=\tilde P_0$).
%}

We aim to complete this  section with some numerical tests for the liquid chartist limit. The asymptotic behavior of the solution of the original model and the reduced model is depicted in Figure \ref{RedC}.  The distance between both models with respect to  $\gamma$ is shown in Figure \ref{ErrorC}. Both numerical tests support our analysis that for $\gamma\to 0$ the Beja-Goldman model converges to the reduced model \eqref{RCL1}. 
Furthermore, in the case $\epsilon<b$ no reduction exists and Figure \ref{RepChart} depicts the repelling slow manifold. Here, we see that the solution of the  formally reduced system remains on the slow manifold while the solution of the original model oscillates with increasing amplitude.

\begin{figure}[h!]
\begin{center}
\includegraphics[width=0.6\textwidth]{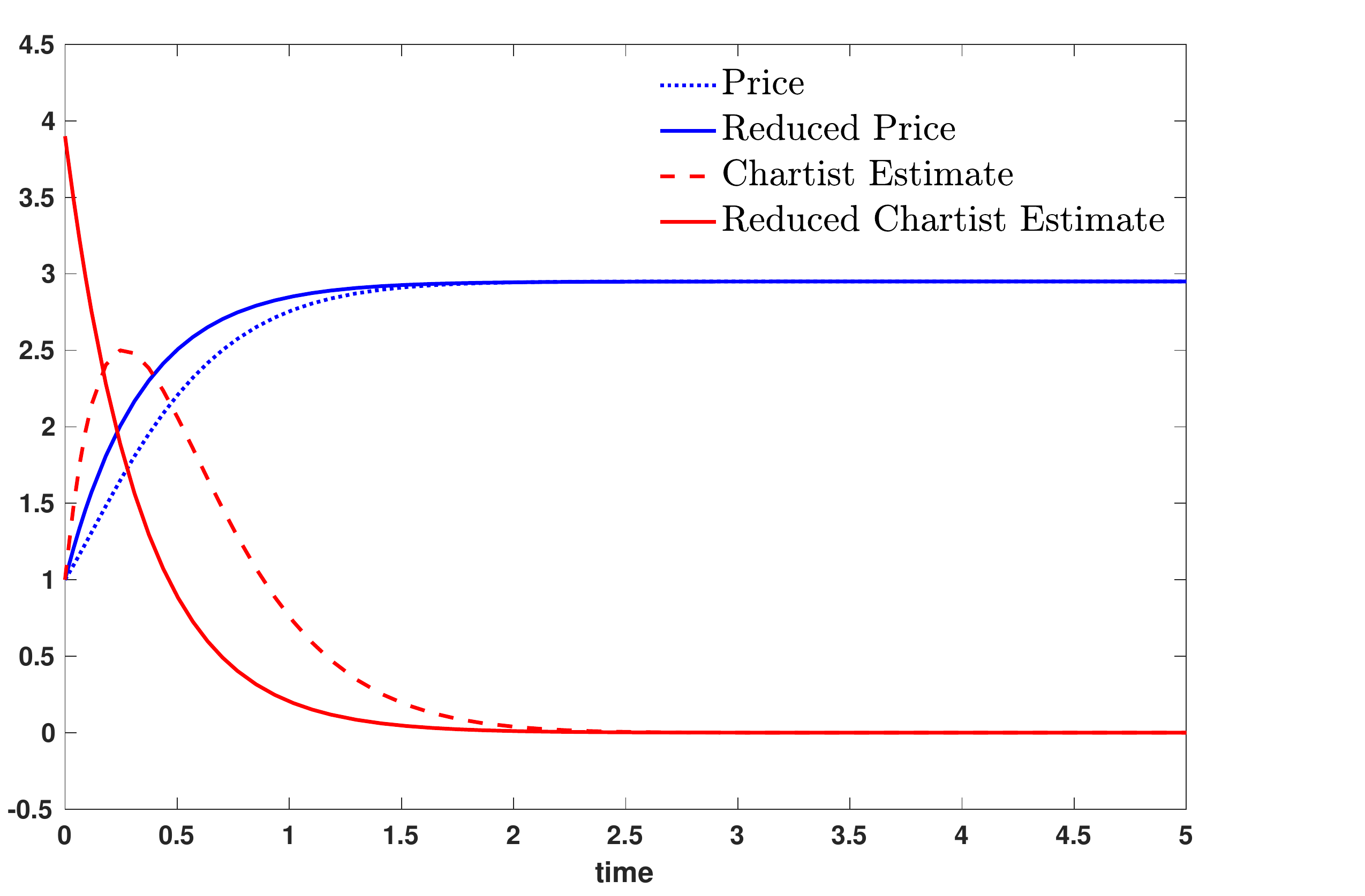}
\end{center}
\caption{Time evolution of the original Beja Goldman model and reduced model in the liquid chartist regime ($\epsilon= 2, a=2, b=1,\gamma=0.1,F=3,r=0.1$). %Initial values of reduced model $\tilde{P}_0=1,\ \tilde{\Psi}_0=3.9$.
 }\label{RedC}
\end{figure}

\begin{figure}[h!]
\begin{center}
\includegraphics[width=0.6\textwidth]{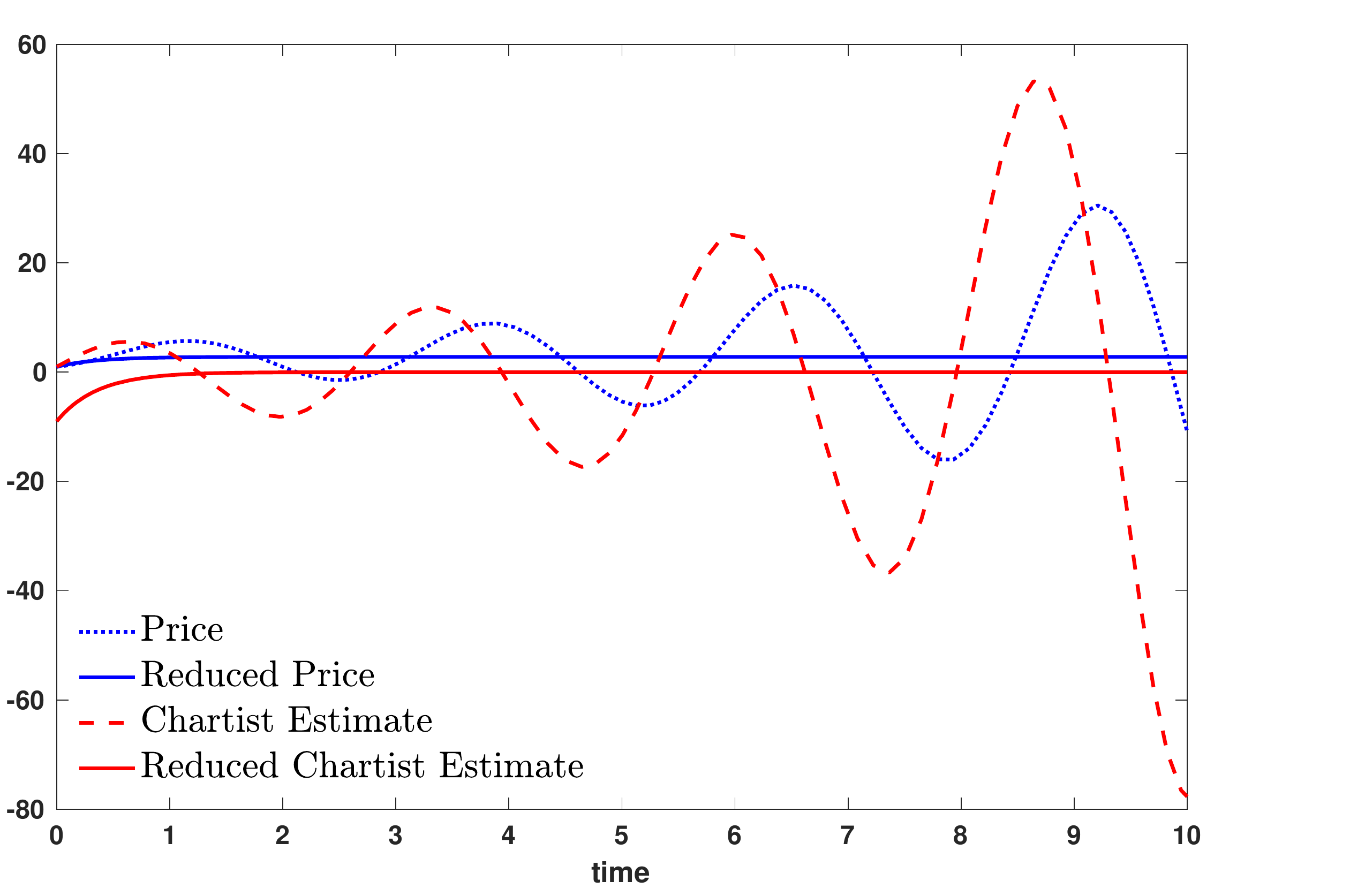}
\end{center}
\caption{Non-stable case with diverging reduced model. Initial values $P_0=1,\ \Psi_0= 1$ are  no elements of the  slow manifold.
Further parameters are given by $\epsilon=1.8,  a=1, b=2,\gamma=0.1,F=3,r=0.1$.
 }\label{RepChart}
\end{figure}

\begin{figure}[h!]
\begin{center}
\includegraphics[width=0.6\textwidth]{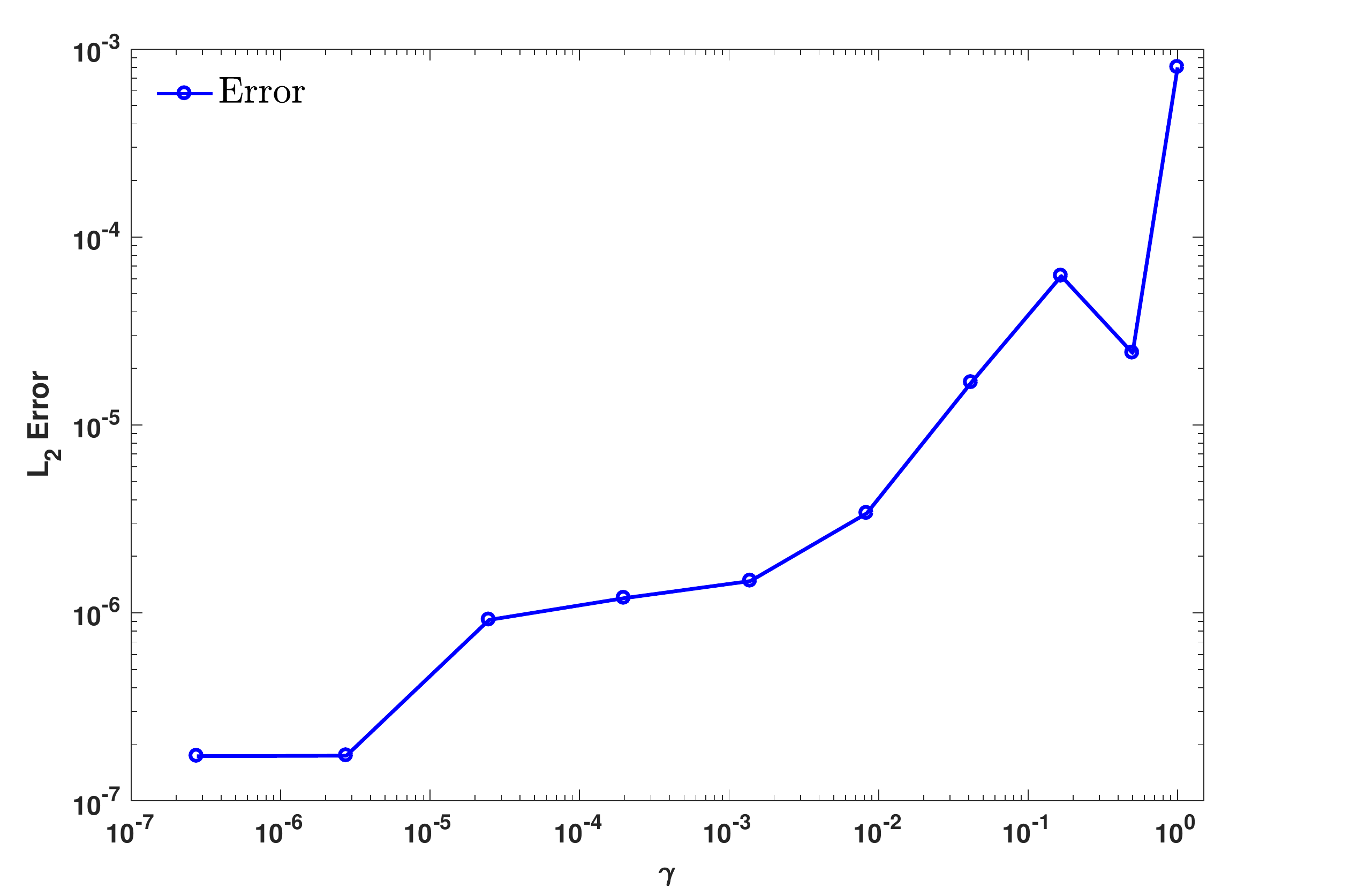}
\end{center}
\caption{The $L_2$ distance at time $t=10$ of the original Beja-Goldman model and the reduced model for different values $\gamma$. The other parameters are chosen to be: $ a=2, b=1,\epsilon=3,F=3,r=0.1$}\label{ErrorC}
\end{figure}

\clearpage

\section{Economic Implications and Conclusion }
Before we discuss the liquid market and liquid chartist limit we want to quickly summarize the behavior of the original Beja-Goldman model in the standard parameter regime.
The stability behavior depends on an interplay of the parameters $a,b,\epsilon,\gamma$. In particular, the system is always stable if the strength of chartist $b$ is below the inverse market depth $\epsilon$ i.e. $b<\epsilon$. 
If no chartists are present $(b=0)$ the system is always stable. Increasing the impact of chartist leads after a certain threshold to stable oscillatory behavior, then to unstable oscillatory behavior and to a blow up for sufficiently large $b$. In fact, not only the impact of chartists, but especially the speed of the adjustment of their estimate $\frac{1}{\gamma}$ heavily influences the price behavior.  \\ \\
In case of a valid reduction we expect that the behavior of the original system is well approximated by the corresponding reduction. Thus, in summary the model characteristics of the reduced model in the 
liquid market limit and liquid chartist limit are given by:
\begin{itemize}
\item \textbf{Liquid Market Limit} In the equilibrium market case, which corresponds to the limit of an infinite market depth, the reduced system only depends on the parameter $a,b,\gamma$. The system is stable and is a valid reducton of the original system if the strength of fundamentalists $a$ is larger than the strength of chartists $b$ times the speed of the adjustment of the estimate $\frac{1}{\gamma}$ i.e. $a>\frac{b}{\gamma}$. Thus, the price converges to the value $P_{\infty}=F-r\frac{b}{a}$ and the chartist estimate to zero. 
\item \textbf{Liquid Chartist Limit}
The behavior of the reduced system in the liquid chartist case is very similar to the liquid market case. The reduced model depends on the parameter $a,b, \epsilon$.  The reduced system exists and  is stable if the inverse market depth $\epsilon$ is larger that the strength of chartists $b$ i.e. $\epsilon>b$. The asymptotic equilibrium price is given by $P_{\infty}=F-r\frac{b}{a}$ and the chartist estimate converges to zero.
 \end{itemize}
The previously described market regimes are only present in extreme situations i.e. in case of an infinite market depth or in case of infinitely fast high frequency trader. 
Unfortunately, the reduced models are only a valid approximation of the original model for special sets of parameters. As discussed previously, the reduction in the liquid market limit (liquid chartist limit) is only a reduction of the original model for $a>\frac{b}{\gamma}$ ($\epsilon>b$).  Therefore one may ask critically what the benefits of  Tikhonov-Fenichel reductions are. In general, we habe 3 main advantages:
\begin{itemize}
\item The reductions replicate original model behavior.
\item The reductions are simpler to analyze than the original systems (as the dimension is always lower).
\item The reductions make it possible to study the original system even if no explicit solution of the original model exists.
\end{itemize}
In cases where the original model can be solved explicitly, Tikhonov-Fenichel reductions still can be applied but one does not obtain novel insights of the original model. 
Therefore this study has to be seen as a novel example to utilize the asymptotic theory, although we do not gain new information on the Beja-Goldman model. 
This work rather shows the novel interpretation of a financial market model as a singular perturbation problem and exemplified we study the Tikhonov-Fenichel reductions. 
Thus, we have derived the reduced system of the Beja-Goldman model for two asymptotic limits, the liquid market limit and the liquid chartist limit.
Furthermore, we have analyzed the behavior of these reduced systems and have verified our analysis by several numerical test. 
Especially the simple analysis of the reduced systems which perfectly replicate the asymptotic behavior of the original model reveals the importance of Tikhonov-Fenichel reductions.\\ \\
Finally, we want to point out that such an analysis can be extended to more general disequilibirum models provided they are in singular perturbation form. 
This is of major importance in the case that the original model is more complicated (i.e. of higher dimension) and the asymptotic behavior cannot be studied directly.

\section*{Acknowledgement}
T. Trimborn was funded by the Deutsche Forschungsgemeinschaft (DFG, German Research Foundation) under Germany's Excellence Strategy – EXC-2023 Internet of Production – 390621612.\\
T. Trimborn gratefully acknowledges support by the Hans-Böckler-Stiftung and the RWTH Aachen University Start-Up grant. 
T. Trimborn acknowledges the support by the ERS Prep Fund - Simulation and Data Science. 
The work was partially funded by the Excellence Initiative of the German federal and state governments.

\clearpage

%-- LITERATUR ----------------------------------------------------------%
	\clearpage
	\bibliography{literature.bib}
		\bibliographystyle{abbrv}

\end{document}